\documentclass[12pt]{article}
\usepackage{amsmath,amsthm,epsfig,amsfonts,color, bm, url,hhline,subcaption,natbib}
\usepackage{rotating}
\allowdisplaybreaks
\usepackage{authblk}
\usepackage{caption}
\usepackage{comment}
\usepackage{booktabs}
\usepackage{setspace}
\usepackage{array}
\usepackage{indentfirst}
\newtheorem{theorem}{{Theorem}}
\newtheorem{lemma}{{Lemma}}
\newtheorem{remark}{{Remark}}

\def\marginnote#1{\setbox0=\vtop{\hsize4pc
\small\raggedright\noindent\baselineskip9pt \rightskip=0.5pc plus
1.5pc #1}\leavevmode \vadjust{\dimen0=\dp0
\kern-\ht0\hbox{\kern-4.00pc\box0}\kern-\dimen0}}
\def\lboxit#1{\vbox{\hrule\hbox{\vrule\kern6pt
\vbox{\kern6pt#1\kern6pt}\kern6pt\vrule}\hrule}}

\usepackage{xr}

\begin{document}
	\thispagestyle{empty}
	\title{Diagnostic Tests Before Modeling Longitudinal Actuarial Data}
	\author[a]{Yinhuan Li}
	\author[b]{Tsz Chai Fung}
	\author[b]{Liang Peng}
	\author[a]{Linyi Qian\thanks{Corresponding author. Email: lyqian@stat.ecnu.edu.cn}}
	\affil[a]{\footnotesize \emph{Key Laboratory of Advanced Theory and Application in Statistics and Data Science-MOE, School of Statistics, East China Normal University, Shanghai 200241, China}}
	\affil[b]{\footnotesize \emph{Department of Risk Management and Insurance, Georgia State University, Atlanta, GA 30303, USA}}
	\date{}
	\maketitle
	\begin{abstract}
	In non-life insurance, it is essential to understand the serial dynamics and dependence structure of the longitudinal insurance data before using them. Existing actuarial literature primarily focuses on modeling, which typically assumes a lack of serial dynamics and a pre-specified dependence structure of claims across multiple years. To fill in the research gap, we develop two diagnostic tests, namely the serial dynamic test and correlation test, to assess the appropriateness of these assumptions and provide justifiable modeling directions. The tests involve the following ingredients: i) computing the change of the cross-sectional estimated parameters under a logistic regression model and the empirical residual correlations of the claim occurrence indicators across time, which serve as the indications to detect serial dynamics; ii) quantifying estimation uncertainty using the randomly weighted bootstrap approach; iii) developing asymptotic theories to construct proper test statistics. The proposed tests are examined by simulated data and applied to two non-life insurance datasets, revealing that the two datasets behave differently.
	\end{abstract}
	\newpage

\section{Introduction}
In non-life insurance, understanding longitudinal insurance claim datasets is essential yet challenging for various actuarial applications, including ratemaking and risk management. Insurance companies often obtain each policyholder's claim information for multiple contract years, resulting in a longitudinal data structure. Since the insurance claims may be serially correlated due to the unobserved policyholder risk characteristics, contracts are priced based not only on the observed policyholder information but also on the past claim history, referred to as \textit{a posteriori} ratemaking. Existing actuarial literature primarily focuses on the advancements of new statistical models for longitudinal data, which mainly include the copula, random effects, and mixture models:
\begin{itemize}
\item Copula models capture the intertemporal claim dependence for a policyholder by incorporating a copula function to link the marginal claim distributions for each year. Noticeable contributions include \cite{frees2006copula}, \cite{boucher2008models}, \cite{shi2014longitudinal}, \cite{FreesLeeYang2016}, \cite{shi2016multilevel}, \cite{shi2018pair}, \cite{LeeShi2019}, and \cite{yang2019multiperil}.
\item Random effects models capture the unobserved heterogeneities among policyholders via a policyholder-level latent variable, such that the past claim history may influence the predicted future claim. It is first studied in actuarial longitudinal data by \cite{boucher2006fixed} and further expanded by, e.g., \cite{pechon2018multivariate}, \cite{jeong2020predictive}, \cite{oh2020bonus}, \cite{jeong2021multi}, \cite{oh2021multi}, and \cite{tseung2022posteriori}.
\item Mixture models classify policyholders into various risk levels, introducing the dependence of claims over time. Related works are \cite{tzougas2014optimal} and \cite{tzougas2021multivariate}.
\end{itemize}

As non-life insurance is typically a short-term product, the longitudinal actuarial datasets often exist across only a few years (\cite{FreesLeeYang2016}), introducing difficulty in modeling and validating their time-series dynamics. Therefore, it can be seen that most, if not all, of the above research works make the following assumptions:

\begin{itemize}
\item{\it (A1)} There are no serial dynamics on the model parameters, i.e., the marginal claim distribution conditioned on the covariates and the covariates' influence do not change over time $t$.
\item {\it (A2)} The serial dependence of claims follows a pre-specified pattern. For example, Gaussian copula models (\cite{shi2014longitudinal} and \cite{yang2019multiperil}) commonly have three specifications: AR(1) for decaying correlations over lag $l$, exchangeable structure for uniform correlations across $l$, and Toeplitz structure for non-zero correlations up to a fixed $l$; random effects models often assume uniform correlations.
\end{itemize}

Nonetheless, these assumptions may be questionable when we reasonably argue that the covariates cannot fully explain the claim dynamics. For example, the altering macro-environmental status, influenced by social-economic conditions, climate events, underwriting policies, and governmental measures, may simultaneously affect the risk levels of all policyholders, creating a heterogeneity of the conditional claim distributions and dependence structure over time. Violation of the assumptions implies that the models in the existing literature are misspecified, and hence the resulting future claim predictions may be misleading. As a result, it is vital to statistically test the assumptions above before modeling the dataset.

Motivated by the issue above, this paper focuses on developing two diagnostic tests to assess if assumptions (A1) and (A2) are satisfied in modeling the longitudinal claim data. To better appreciate our study, as a starting point, we first model the conditional claim occurrence probabilities by logistic regression, the first step of the ratemaking process in the existing literature, see \cite{Heras2018} and \cite{Kang2020b}. 
Then, we test the serial dynamics and dependence structure of conditional claim occurrence indicators over time. A similar analysis for other parts of a claim distribution, e.g., tail risk measures, is much involved and beyond this paper's scope. We will develop the following tests:
\begin{itemize}
\item \emph{Serial dynamic test}: Identify any structural change of conditional marginal claim occurrence probabilities over time that violates (A1).
\item \emph{Correlation test}: 
Examine if the claim occurrence indicators may be conditionally independent with any or all time lags $l$ to support the study (A2).
\end{itemize}
Rejecting the constant dynamic assumption in (A1) calls for a more delicate model for time series dynamics. However, given a short observable window for actuarial data, it is infeasible to validate any imposed non-constant dynamic structure appropriately, as a nonparametric inference is generally needed. On the other hand, a failure to reject (A1) means the use of longitudinal actuarial data is only for improving inference efficiency, as the future conditional claim probability has the same structure as the previous years. 

In either case, the claim occurrence indicators may be serially dependent across years. Hence, a proper specification of the serial dependence structure under (A2) is mandatory, and the correlation test will offer dependence modeling guidance.

We organize the paper as follows. Section \ref{sec:problem} provides the complete mathematical framework for the research problem stated above. Section \ref{sec:method} presents the methodology and asymptotic results. Section \ref{sec:data} analyzes two public actuarial datasets: the local government property insurance fund data from the state of Wisconsin and the French private motor dataset, which exhibit different behaviors. A simulation study is given in Section \ref{sec:simu} to evaluate the finite sample performance. Section \ref{sec:conclusion} concludes. All proofs are put in Section \ref{sec:proofs}.

\section{Research problem statement} \label{sec:problem}
Following the panel data structure studied in the actuarial literature, we suppose that the actuarial dataset over multiple years is $(Z_{i,t},\boldsymbol{X}_{i,t})$ for $i\in A_t$ and $t=1,\ldots,T$, where $Z_{i,t}\in\{0,1\}$ and $\boldsymbol{X}_{i,t}\in\mathbb{R}^P$ are respectively the claim occurrence indicator and the covariate vector for the $i$-th policyholder at a particular year $t$, $A_t\subseteq \{1,\ldots,n\}:=A$ is a set of policyholders being exposed in year $t$, $n$ is the total number of policyholders in the insurance portfolio, and $T$ is the total number of contract years observed. $Z_{i,t}=1$ if there is at least one claim for policyholder $i$ at time $t$, and $Z_{i,t}=0$ otherwise. Because a policyholder may not renew his/her policy for some years, it is possible that $A_t\neq A_s$ for some $t\neq s$. Hence, the number of policies in year $t$ may not equal that in year $s$, implying that the panel data is imbalanced. 

When one is interested in forecasting the risk of the aggregate loss of a new policyholder, it is critical to model the conditional probability of nonzero claims given the covariates, i.e., $p_{i,t}=P(Z_{i,t}=1|\boldsymbol{X}_{i,t})$. Logistic regression is commonly employed by assuming that the conditional random variable $Z_{i,t}$ given $\boldsymbol{X}_{i,t}$ has a Bernoulli distribution with the claim (occurrence) probability $p_{i,t}:=p_{t}(\boldsymbol{X}_{i,t})$ modeled by a logit link:
\begin{equation}\label{mod}
Z_{i,t}|\boldsymbol{X}_{i,t} \sim \text{Bernoulli}(p_{i,t}),\quad \log\frac{p_{i,t}}{1-p_{i,t}}=\alpha_t+\boldsymbol{\beta}^{\top}_t\boldsymbol{X}_{i,t},
\end{equation}
where $\alpha_t\in \mathbb{R}$ and $\boldsymbol{\beta}_t\in\mathbb{R}^{P}$ are the regression parameters.

\begin{remark}
While the claim probability contains substantial information for the claim behaviors of the policyholder, it is insufficient for ratemaking purposes, which requires full specifications on both claim frequency and severity. For forecasting the conditional claim frequency, one needs to specify a discrete distribution, such as Poisson and Negative Binomial, which affects the variance of the claim frequency and hence may influence the ratemaking decision. For forecasting the conditional Value-at-Risk of the aggregate loss, one can employ quantile regression at an adjusted risk level estimated from the above logistic regression; see  \cite{Kudryavtsev2009},
\cite{Heras2018}, \cite{Kang2020b}, and \cite{Kang2020a}. As the main goal of this paper is to provide a fundamental starting point to test the serial dynamic and dependence assumptions, we refrain from fully specifying the claim frequency and severity distributions.
\end{remark}

On the other hand, if one wants to forecast the conditional claim probability for a future year, it becomes necessary to understand the time dynamics using actuarial data over multiple years. That is, 
we want to forecast the conditional probability $P(Z_{i,T+1}=1|\boldsymbol{X}_{i,T+1})$, which needs to model the dynamic structure of $P(Z_{i,t}=1|\boldsymbol{X}_{i,t})$ over time, i.e., modeling $\alpha_t$ and $\boldsymbol{\beta}_t$ in (\ref{mod}).
When a parametric form is imposed for $\alpha_t$ and $\boldsymbol{\beta}_t$, say, a linear form over time, validation often requires a nonparametric inference and becomes infeasible because $T$ is usually small.  A common technique to overcome this validation challenge assumes (A1), where there are no time dynamics. In the context of the first step of the ratemaking process (i.e., modeling $Z_{i,t}|\boldsymbol{X}_{i,t}$), (A1) is equivalent to the following null hypothesis: 
\begin{equation}\label{null1}
 H_0: \alpha_1=\cdots=\alpha_T=\alpha~\text{and}~ \boldsymbol{\beta}_1=\cdots=\boldsymbol{\beta}_T=\boldsymbol{\beta}.\end{equation}
Under this setting, using actuarial data over multiple years only improves inference efficiency and forecast accuracy rather than modeling time dynamics.  

This paper focuses on developing a diagnostic test for (\ref{null1}) in modeling conditional claim occurrence probability. A similar study for risk forecast is much involved and will be investigated separately.
More specifically,  Section \ref{sec:method:test1} develops a diagnostic test, called the \emph{serial dynamic test}, for a constant dynamic in (\ref{null1}), which can be employed before using actuarial data over multiple years. In this test, we assume that policyholders are independent, i.e., $(Z_{i,t}, \boldsymbol{X}_{i,t})$ and $(Z_{j,t}, \boldsymbol{X}_{j,t})$ are independent for $i\neq j$. However, to reflect that the claim occurrences of a policyholder are serially dependent as extensively modeled by the literature, we allow for the dependence between the conditional variable of $Z_{i,t}$ given $\boldsymbol{X}_{i,t}$ and that of $Z_{i,s}$ given $\boldsymbol{X}_{i,s}$ for $t\neq s$. The dependence between $\boldsymbol{X}_{i,t}$ and $\boldsymbol{X}_{i,s}$ is also allowed. Further, the distribution of $\boldsymbol{X}_{i,t}$ may be different from that of $\boldsymbol{X}_{i,s}$ for $t\neq s$ because the insurance portfolio may shift over time, e.g., policyholders get older over time.
To construct a Hotelling's T-test, we adopt the random weighted bootstrap method in \cite{JinYingWei2001} and \cite{Zhu2016} to quantify the uncertainty of maximum likelihood estimation for the logistic regressions, which is better for our imbalanced longitudinal data.


When the test does not reject the null hypothesis (\ref{null1}), it may be justifiable to model the longitudinal data without incorporating any time dynamics as in the literature, and the past data can be used to improve the inference for parameters and get better future claim occurrence predictions. 

Regardless of the existence of non-stationary dynamics for the conditional claim occurrence probability over time, the claim occurrence may still be serially dependent across years. As a result, it is also crucial to understand the serial dependence between $Z_{i,t}|\boldsymbol{X}_{i,t}$ and $Z_{i,t+l}|\boldsymbol{X}_{i,t+l}$ for $t=1,\ldots,T-l$ and $l=1,\ldots, T-1$, which helps identify the appropriate class of dependence models. Here, we propose to test the conditional independence of the claims given by
\begin{equation}\label{null_indep}
H_0: Z_{i,1},\ldots,Z_{i,T} ~\text{are jointly independent conditioned on}~\bm{X}_{i,1},\ldots,\bm{X}_{i,T}
\end{equation}
for a full comparison of all $T$ years, or
\begin{equation}\label{null_indep_pair}
H_0: Z_{i,t}~\text{and}~Z_{i,t+l}~\text{are independent conditioned on}~\bm{X}_{i,t}~\text{and}~\bm{X}_{i,t+l}
\end{equation}
for a pairwise comparison with a time lag of $l$. In Section \ref{sec:method:test2}, we will develop a \emph{correlation test} for the conditional independence null hypothesis (\ref{null_indep}) or (\ref{null_indep_pair}). The diagnostic results will provide some useful dependence modeling guidance.
\begin{itemize}
\item If (\ref{null_indep}) is not rejected, then it suffices to assume conditional independence among all observations, and an ordinary regression model may be more parsimonious and powerful in prediction than any dependence models in the literature.
\item If (\ref{null_indep_pair}) is rejected for a small but not large $l$, then short-term memory structures, including the AR(1) and Toeplitz structures, will be more suitable in capturing the serial dependence.
\item If (\ref{null_indep_pair}) is rejected even for a large $l$, a long-term memory structure like a uniform correlation may be needed.
\end{itemize}

\section{Methodology and asymptotic results}	\label{sec:method}

In this section, we develop two diagnostic tests, \emph{serial dynamic test} and \emph{correlation test}, for longitudinal actuarial data. Throughout this paper, we define $n_t$ and $n_{s,t}$ as the number of elements in $A_t$ and $A_s\cap A_t$, respectively. Hence, we have $n_{t,t}=n_t$. Also, put $\boldsymbol{\bar{X}}_{i,t}=\left( 1,\boldsymbol{X}_{i,t}^\top\right) ^\top$, where $\top$ denotes the transpose of a matrix or vector.

\subsection{Serial dynamic test} \label{sec:method:test1}
Suppose that model (\ref{mod}) holds. Since we do not want to specify a particular dependence structure among $Z_{i,1},\cdots,Z_{i,T}$ given $\boldsymbol{X}_{i,1},\ldots,\boldsymbol{X}_{i,T}$, we do not estimate $\alpha_t$'s and $\boldsymbol{\beta}_t$'s jointly. Instead, we use the $t$-th year's data to estimate $\boldsymbol{\gamma}_t=(\alpha_t, \boldsymbol{\beta}_t^{\top})^{\top}$ by the logistic regression estimation
	\begin{equation}\label{lre}
		\hat{\boldsymbol{\gamma}}_t=\arg\max\sum_{i\in A_t}\{Z_{i,t}\log(p_{i,t})+(1-Z_{i,t})\log(1-p_{i,t})\}.
	\end{equation}
Since (\ref{null1}) is equivalent to	$H_0: \boldsymbol{\gamma}=\boldsymbol{0}$ with
	$\boldsymbol{\gamma}=(\boldsymbol{\gamma}_2^{\top}-\boldsymbol{\gamma}_1^{\top},\cdots,\boldsymbol{\gamma}_T^{\top}-\boldsymbol{\gamma}_1^{\top})^{\top}$, we estimate $\boldsymbol{\gamma}$ by $\hat{\boldsymbol{\gamma}}=(\hat{\boldsymbol{\gamma}}_2^{\top}-\hat{\boldsymbol{\gamma}}_1^{\top},\cdots,\hat{\boldsymbol{\gamma}}_T^{\top}-\hat{\boldsymbol{\gamma}}_1^{\top})^{\top}$. To derive the asymptotic limit of $\hat{\boldsymbol{\gamma}}$, we employ the following regularity conditions.
\begin{itemize}
\item{\it (C1)} $(Z_{i,s},\boldsymbol{X}_{i,s}^{\top})^{\top}$ and $(Z_{j,t}, \boldsymbol{X}_{j,t}^{\top})^{\top}$ are independent  when $i\neq j$ for any $s,t=1,\dots,T$. $Z_{i,t}|(\boldsymbol{X}_{i,s},\boldsymbol{X}_{i,t})$ has the same distribution as  $Z_{i,t}|\boldsymbol{X}_{i,t}$  if $i\in A_s\cap A_t$ for any $s,t=1,\dots,T$.
\item {\it (C2)}
		For any fixed $t=1,\cdots,T$, \textnormal{$\left\lbrace (Z_{i,t},\boldsymbol{X}_{i,t}^{\top})^{\top}\right\rbrace_{i=1,\ldots,n_t} $} is a sequence of independent and identically distributed random vectors. 
\item {\it (C3)} $E(||\boldsymbol{X}_{i}||^{2+\delta})<\infty$ for some $\delta>0$, where $\boldsymbol{X}_i=(\boldsymbol{X}_{i,1}^\top,\dots,\boldsymbol{X}_{i,T}^\top)^\top$. 
\item{\it (C4)}
	\textnormal {$\Sigma_t=E\left\lbrace p_{i,t}(1-p_{i,t})\bar{\boldsymbol{X}}_{i,t}\bar{\boldsymbol{X}}_{i,t}^\top\right\rbrace $ is positive definite for $t=1,\dots,T$.}
\item {\it (C5)}
		\textnormal{$n_{s,t}/n\to a_{s,t}\in[0, 1]$ as $n\to \infty$ for $s, t=1,\dots,T$ with $a_{s,s}=a_s>0$.}
\end{itemize}

	\begin{theorem}\label{Th1}
		Under conditions (C1)-(C5) and model (\ref{mod}), we have
		\[\sqrt n (\hat{\boldsymbol{\gamma}}-\boldsymbol{\gamma})=(a_2^{-1}\left( \Sigma_2^{-1}\boldsymbol{W}_2\right) ^\top- a_1^{-1}\left(\Sigma_1^{-1}\boldsymbol{W}_1\right)^\top,\cdots,a_T^{-1} \left( \Sigma_T^{-1}\boldsymbol{W}_T\right) ^\top-a_1^{-1}\left( \Sigma_1^{-1}\boldsymbol{W}_1\right) ^\top)^{\top}+o_p(1)\]
		as $n\to\infty$,
		where the joint normal limit of $\boldsymbol{W}_1,\cdots,\boldsymbol{W}_T$ has mean 0 and the following covariance 
		\[E\{\boldsymbol{W}_s\boldsymbol{W}_t^{\top}\}=a_{s,t}E\{(p_{i,s,t}-p_{i,s}p_{i,t})\bar{\boldsymbol{X}}_{i,s}\bar{\boldsymbol{X}}_{i,t}^{\top}\}\]
		with $p_{i,s,t}=P(Z_{i,s}=1,Z_{i,t}=1|\boldsymbol{X}_{i,s},\boldsymbol{X}_{i,t})$ and $p_{i,s,s}=p_{i,s}$. We denote  the asymptotic covariance of $\hat{\boldsymbol{\gamma}}$ as $\Sigma$.
	\end{theorem}
	
	To estimate the asymptotic covariance of $\hat{\boldsymbol{\gamma}}$, we adopt the random weighted bootstrap method in \cite{JinYingWei2001} and \cite{Zhu2016}  as follows.
	\begin{itemize}
		\item Step \romannumeral1 1) Draw a random sample with size $n$ from a distribution with mean one and variance one, say the standard exponential distribution. Denote them by $\{\delta_i^b\}_{i=1}^n$.
		\item Step \romannumeral1 2) Solve
\begin{equation} \label{eq:gamma_boot}
\hat{\boldsymbol{\gamma}}_t^b=\arg\max\sum_{i\in A_t}\delta_i^b\{Z_{i,t}\log(p_{i,t})+(1-Z_{i,t})\log(1-p_{i,t})\}
\end{equation}
		and write $\hat{\boldsymbol{\gamma}}^b=(\hat{\boldsymbol{\gamma}}_2^{b\top}-\hat{\boldsymbol{\gamma}}_1^{b\top},\cdots,\hat{\boldsymbol{\gamma}}_T^{b\top}-\hat{\boldsymbol{\gamma}}_1^{b\top})^{\top}.$
		\item Step \romannumeral1 3) Repeat the above two steps $B$ times to get $\{\hat{\boldsymbol{\gamma}}^b\}_{b=1}^B$. 
	\end{itemize}
	
\begin{theorem}\label{Th2}
Under the conditions of Theorem \ref{Th1},
$\sqrt n (\hat{\boldsymbol{\gamma}}^b-\hat{\boldsymbol{\gamma}})$ and $\sqrt n(\hat{\boldsymbol{\gamma}}-\boldsymbol{\gamma})$ have the same normal limit as $n\to\infty$.	
\end{theorem}

From the theorem above, we estimate the asymptotic covariance of $\sqrt n(\hat{\boldsymbol{\gamma}}-\boldsymbol{\gamma})$ by
\[\hat\Sigma=\frac nB\sum_{b=1}^B(\hat{\boldsymbol{\gamma}}^b-\hat{\boldsymbol{\gamma}})(\hat{\boldsymbol{\gamma}}^b-\hat{\boldsymbol{\gamma}})^{\top},\]
and define the \emph{aggregated serial dynamic test statistic} for (\ref{null1}) by
\begin{equation} \label{eq:stat:serial_agg}
\Delta_n=n\hat{\boldsymbol{\gamma}}^{\top}\hat\Sigma^{-1}\hat{\boldsymbol{\gamma}}.
\end{equation}
	
\begin{theorem}\label{Th3}
Under the conditions of Theorem \ref{Th1} and the null hypothesis (\ref{null1}), 
\[\Delta_n\overset{d}{\to}\chi^2((T-1)\times (P+1))~\text{as}~ n\to\infty ~\text{and}~ B \to\infty.\]
\end{theorem}
	
From the theorem above, we reject the null hypothesis of (\ref{null1}) at level $a$ whenever $\Delta_n>\chi^2_{1-a,(T-1)\times (P+1)}$, where $\chi^2_{1-a,(T-1)\times (P+1)}$ denotes the $(1-a)$-quantile of a chi-squared distribution with $(T-1)\times (P+1)$ degrees of freedom. Alternatively, one can develop a \emph{pairwise} serial dynamic test to detect the structural difference of conditional claim distributions between years $s$ and $t$ with a test statistic
\begin{equation} \label{eq:stat:serial_pair}
\Delta_{s,t,n}=n(\hat{\boldsymbol{\gamma}}_s-\hat{\boldsymbol{\gamma}}_t)^{\top}\hat\Sigma_{s,t}^{-1}(\hat{\boldsymbol{\gamma}}_s-\hat{\boldsymbol{\gamma}}_t),
\end{equation}
where
\[\hat\Sigma_{s,t}=\frac{n}{B}\sum_{b=1}^{B}\left\lbrace (\hat{\boldsymbol{\gamma}}^b_s-\hat{\boldsymbol{\gamma}}^b_t)-(\hat{\boldsymbol{\gamma}}_s-\hat{\boldsymbol{\gamma}}_t)\right\rbrace \left\lbrace (\hat{\boldsymbol{\gamma}}^b_s-\hat{\boldsymbol{\gamma}}^b_t)-(\hat{\boldsymbol{\gamma}}_s-\hat{\boldsymbol{\gamma}}_t)\right\rbrace ^{\top}.\]
It follows from Theorems \ref{Th1} to \ref{Th3} that $\Delta_{s,t,n}\overset{d}{\to}\chi^2(P+1)$ as $n\to\infty$ and $B \to\infty$, which can be employed to test for no serial change between years $s$ and $t$ as above.
\begin{remark}
	The proposed serial dynamic test is different from a standard Wald test in two perspectives. Firstly, the serial dynamic test caters to imbalanced longitudinal claim data, i.e., $A_s\neq A_t$ for $s\neq t$, which is not treated in a standard Wald test. Secondly, a standard Wald test requires estimating parameters $(\bm{\gamma}_1^\top,\dots,\bm{\gamma}_T^\top)^\top$ by a conditional joint likelihood function, which requires a specification of the dependence structure of $\{Z_{i,t}|\bm{X}_{i,t}\}_{t=1}^{T}$ over time. On the other hand, the proposed serial dynamic test estimates $\bm{\gamma}_1,\dots,\bm{\gamma}_T$ by maximizing the conditional marginal likelihood function separately for each $t=1,\ldots,T$ without putting any restrictions on the serial dependence. This minimizes assumptions and reduces the computational burden.
\end{remark}
\subsection{Correlation test} \label{sec:method:test2}
To test the null hypotheses (\ref{null_indep}) and (\ref{null_indep_pair}), we first construct the empirical correlation of the residuals of the claim indicators between time $s$ and $t$:
\begin{equation} \label{eq:corr_pair}
	\hat{\rho}_{s,t}=\frac{1}{n_{s,t}}\sum_{i\in A_{s}\cap A_t}\left(\frac{Z_{i,s}-\hat{p}_{i,s}}{\sqrt{\hat{p}_{i,s}(1-\hat{p}_{i,s})}}\right)\left(\frac{Z_{i,t}-\hat{p}_{i,t}}{\sqrt{\hat{p}_{i,t}(1-\hat{p}_{i,t})}}\right),
\end{equation}
where $\hat{p}_{i,t}=\exp( \hat{\alpha}_t+\hat{\boldsymbol{\beta}}_t^{\top}\boldsymbol{X}_{i,t}) /\{1+\exp( \hat{\alpha}_t+\hat{\boldsymbol{\beta}}_t^{\top}\boldsymbol{X}_{i,t}) \}$ is the predicted claim probability under model (\ref{mod}), and $(\hat{\alpha}_t,\hat{\bm{\beta}}_t^\top)^\top$ are the estimated parameters from (\ref{lre}).  A large magnitude of $\hat{\rho}_{s,t}$ is an indication that $Z_{i,t}$ and $Z_{i,s}$ are dependent conditional on the covariates. Hence, $\hat{\rho}_{s,t}$ assesses the pairwise dependence. We also denote $\rho_{s,t}=E\left\lbrace Corr(Z_{i,s},Z_{i,t}|\bm{X}_{i,s},\bm{X}_{i,t})\right\rbrace $ as the true expected correlation of the claim indicators. The asymptotic limit of $\hat{\rho}_{s,t}$ further requires the following conditions:
\begin{itemize}
\item {\it (C6)} \textnormal{$a_{s,t}\in (0,1]$, where $a_{s,t}$ is given by condition (C5).}
\item {\it (C7)} \textnormal{$\sup_{(\tilde{\boldsymbol{\gamma}}_1^{\top},\cdots,\tilde{\boldsymbol{\gamma}}_T^{\top})^{\top}\in\Omega}E\left\lbrace \prod_{t=1}^T\exp\left(|\tilde{\boldsymbol{\gamma}}_t^\top\bar{\boldsymbol{X}}_{i,t}|\right)\right\rbrace <\infty$ for some $\Omega$ being a neighborhood of the true value of $(\boldsymbol{\gamma}_1^{\top},\cdots,\boldsymbol{\gamma}_T^{\top})^{\top}$.}
\end{itemize}

\begin{theorem} \label{Th4}
Under conditions (C1)-(C7) and model (\ref{mod}), we have
\[\sqrt{n}(\hat{\rho}_{s,t}-\rho_{s,t})\overset{d}{\to}N(0, \lambda_{s,t}),\]
where $\lambda_{s,t}$ is provided by the proof.
\end{theorem}

Next, we can follow the procedures in Section \ref{sec:method:test1} by applying a random weighted bootstrap method to estimate the asymptotic variance $\lambda_{s,t}$ as follows.
	\begin{itemize}
		\item Step \romannumeral2 1) Use the same random sample $\{\delta_i^b\}_{i=1}^n$ generated by Step a1) of the previous bootstrapping scheme (for the asymptotic covariance of $\hat{\bm{\gamma}}$) as the bootstrap weights of the $n$ observations.
		\item Step \romannumeral2 2) Compute the bootstrapped empirical residual correlation as follows:
		\begin{equation} \label{eq:corr_boot}
			\hat{\rho}^b_{s,t}=\frac{1}{n_{s,t}}\sum_{i\in  A_{s}\cap A_t}\delta_i^b\left(\frac{Z_{i,s}-\hat{p}^b_{i,s}}{\sqrt{\hat{p}^b_{i,s}(1-\hat{p}^b_{i,s})}}\right)\left(\frac{Z_{i,t}-\hat{p}^b_{i,t}}{\sqrt{\hat{p}^b_{i,t}(1-\hat{p}^b_{i,t})}}\right),
		\end{equation}
where $\hat{p}^b_{i,t}=\exp(\hat{\alpha}^b_t+\hat{\boldsymbol{\beta}}_t^{b\top}\boldsymbol{X}_{i,t})/\{ 1+\exp(\hat{\alpha}^b_t+\hat{\boldsymbol{\beta}}_t^{b\top}\boldsymbol{X}_{i,t})\} $ with $\hat{\bm{\gamma}}^b_t=(\hat{\alpha}^b_t,\hat{\bm{\beta}}^b_t)$ given by (\ref{eq:gamma_boot}).
		\item Step \romannumeral2 3) Repeat the above two steps $B$ times to get $\{\hat{\rho}^b_{s,t}\}_{b=1}^B$.
	\end{itemize}

\begin{theorem}\label{Th5}
Under the conditions of Theorem \ref{Th4}, 
$\sqrt n (\hat{\rho}_{s,t}^b-\hat{\rho}_{s,t})$ and $\sqrt n(\hat{\rho}_{s,t}-\rho_{s,t})$ have the same normal limit as $n\to\infty$.
\end{theorem}

From the theorem above, we estimate the asymptotic variance of $\sqrt n(\hat{\rho}_{s,t}-\rho_{s,t})$ 
by
\begin{equation} \label{eq:lambda_hat}
\hat{\lambda}_{s,t}=\frac nB\sum_{b=1}^B(\hat{\rho}^b_{s,t}-\hat{\rho}_{s,t})^2,
\end{equation}
and define the \emph{pairwise correlation test statistic} for null hypothesis of (\ref{null_indep_pair}) as
\begin{equation} \label{eq:test_stat_corr_pair}
\Gamma_{s,t,n}=n\frac{\hat{\rho}_{s,t}^2}{\hat{\lambda}_{s,t}}.
\end{equation}
	
\begin{theorem}\label{Th6}
Under the conditions of Theorem \ref{Th4} and the null hypothesis (\ref{null_indep_pair}), 
\[\Gamma_{s,t,n}\overset{d}{\to}\chi^2(1)~\text{as}~ n\to\infty ~\text{and}~ B \to\infty.\]
\end{theorem}

Hence, we reject the null hypothesis of (\ref{null_indep_pair}) at level $a$ whenever $\Gamma_{s,t,n}>\chi^2_{1-a,1}$.

\begin{remark}\label{rem1}
To assess the joint dependence across all $T$ years, i.e., testing the null hypothesis of (\ref{null_indep}), we also create two $(T(T-1)/2)\times 1$ vectors of empirical and bootstrapped residual correlations $\hat{\bm{\rho}}:=(\{\hat{\rho}_{s,t}\}_{1\leq s<t\leq T})$ and $\hat{\bm{\rho}}^b:=(\{\hat{\rho}_{s,t}^b\}_{1\leq s<t\leq T})$ across all pairs of years. 
Then, we construct an \emph{aggregated correlation test statistic}, given by
\begin{equation} \label{eq:test_stat_corr}
\Gamma_n=n\hat{\bm{\rho}}^{\top}\hat\Lambda^{-1}\hat{\bm{\rho}}\quad\text{with}\quad
\hat{\Lambda}=\frac nB\sum_{b=1}^B(\hat{\bm{\rho}}^b-\hat{\bm{\rho}})(\hat{\bm{\rho}}^b-\hat{\bm{\rho}})^{\top}.
\end{equation}
Following the proof techniques from Theorems \ref{Th4} to \ref{Th6},  we can show that $\Gamma_n\overset{d}{\to}\chi^2(T(T-1)/2)$ as $n\to\infty$ and $B \to\infty$ subject to the following additional regularity condition:
\begin{itemize}
\item {\it (C8)} $n_{t_1,t_2,t_3,t_4}/n\to a_{t_1,t_2,t_3,t_4}\in[0,1]$ for any $t_1,t_2,t_3,t_4=1,\ldots,T$, where $n_{t_1,t_2,t_3,t_4}$ is the number of elements in $A_{t_1}\cap A_{t_2}\cap A_{t_3}\cap A_{t_4}$.
\end{itemize}
Therefore, we reject the null hypothesis of (\ref{null_indep}) at level $a$ when $\Gamma_n>\chi^2_{1-a,T(T-1)/2}$. 
We omit the proof. \end{remark}

\begin{remark}
One may alternatively perform a standard t-test on the residual correlation by applying the {\normalfont \texttt{cor.test}} function in {\normalfont \texttt{R}}. However, our proposed correlation test is different from this standard method in several perspectives. Firstly, our test is applied to imbalanced longitudinal claim data. Secondly, our test can be extended to assess the joint dependence for all $T>2$ years (see Remark \ref{rem1}), while the standard correlation test only provides a pairwise assessment for $T=2$ years. Thirdly, the standard correlation test tends to underestimate the standard error of the residual correlation because it fails to cater for the estimation uncertainty of $\hat{p}_{i,t}$, see Section \ref{sec:simu_corr} for more details.
\end{remark}

\section{Data Analysis}\label{sec:data}
This section applies the developed tests to two public datasets.

\subsection{Wisconsin LGPIF dataset} \label{sec:data_lgpif}
In this subsection, we analyze the Local Government Property Insurance Fund (LGPIF) data from the state of Wisconsin, which is publicly available in \url{https://sites.google.com/a/wisc.edu/jed-frees/home}. The dataset records the claim information of $n=1,234$ local government policyholders (entities) from 2006 to 2010. For each entity, the claim information is recorded on a year-by-year basis. Since not all entities are insured through a full five years from 2006 to 2010, the total number of entity-years is $\tilde{n}=5,677<1,234\times 5$. For each entity year, the claim frequencies and average claim severities are recorded across six types of coverages (perils). 
Explanatory variables accompany each observation; see Table \ref{tab:prelim_x} for the variable descriptions.

\begin{table}[!h]
\centering
\caption{\label{tab:prelim_x}[Wisconsin LGPIF dataset] Descriptive summary for the explanatory variables.}
\resizebox{\textwidth}{!}{%
\begin{tabular}{llll}
\toprule
\multicolumn{1}{c}{Index} & \multicolumn{1}{c}{Variable name} & \multicolumn{1}{c}{Type} & \multicolumn{1}{c}{Description} \\
\midrule
1 & \texttt{TypeCity} & Categorical & Indicator for city entity.\\
2 & \texttt{TypeCounty} & Categorical & Indicator for county entity.\\
3 & \texttt{TypeSchool} & Categorical & Indicator for school entity.\\
4 & \texttt{TypeTown} & Categorical & Indicator for town entity.\\
5 & \texttt{TypeVillage} & Categorical & Indicator for village entity.\\
-- & \texttt{TypeMisc} & Categorical & Indicator for miscellaneous entity (reference group).\\
6 & \texttt{IsRC} & Binary & Indicator for replacement cost.\\
7 & \texttt{log(1+CoverageBC)} & Continuous & Coverage amount (transformed). \\
8 & \texttt{lnDeductBC} & Binary & Deductible amount (transformed).\\
\hhline{====}
\end{tabular}
}
\end{table}
	
For this research problem, we focus only on the building and contents of BC peril because it is the only peril that contains sufficient observations with a non-zero number of claims. Table \ref{tab:data_summary} presents the summary statistics of the empirical observations under the BC peril. We first observe that the number of policyholders $n_t$ slightly decreases as $t$ increases, showing that a few entities drop out of the pool over time. Also, the proportion of observations with non-zero claims, i.e., $\sum_{i=1}^{n_t}Z_{i,t}/n_t$, fluctuates quite substantially over the years. In particular, higher proportions of policyholders filed claims in 2007 and 2010. However, solely based on this information, one cannot conclude that the distribution of $Z_{i,t}$ given $\bm{X}_{i,t}$ exhibits serial dynamic over time (i.e., violation of (A1)) because the shift of some variables over time, such as \texttt{log(1+CoverageBC)}, may explain well such a dynamic. Therefore, we need to understand if the covariates fully explain such a fluctuation or if the conditional claim probability structurally changes over time, which may imply a potentially non-constant time-series structure underlying the claim arrival process that can hardly be predicted with only five years of claim experience.
	
	\begin{table}[!h]
		\centering
		\caption{[Wisconsin LGPIF dataset] Summary statistics of the observations (BC peril) across different years.}
			\begin{tabular}{lrrrrr}
				\toprule
				Year & \multicolumn{1}{c}{2006} & \multicolumn{1}{c}{2007} & \multicolumn{1}{c}{2008} & \multicolumn{1}{c}{2009} & \multicolumn{1}{c}{2010} \\ \hline
				Number of observations & 1162 & 1147 & 1134 & 1117 & 1117 \\
				Proportion of non-zero claims & 0.2659 & 0.3112 & 0.2734 & 0.2695 & 0.3644 \\
				\hhline{======}
			\end{tabular}
		\label{tab:data_summary}
	\end{table}
	
For each year $t$, a set of observations $\{(Z_{i,t},\bm{X}_{i,t})\}_{i=1,\ldots,n_t}$ is fitted to a logistic regression model. The estimated parameters and their standard errors across different years are presented in Table \ref{tab:reg_coef}. The variables \texttt{TypeCity}, \texttt{TypeCounty}, \texttt{log(1+CoverageBC)}, and \texttt{lnDeductBC} are significantly non-zero across (almost) all years. 
For these variables, the signs of the regression coefficients do not change over time, providing evidence that the fitted models do not drastically change over time. On the other hand, while the intercept parameter is insignificant, it varies substantially across years: the intercept is positive for 2007 and 2010 and negative for other years. This result echoes Table \ref{tab:data_summary} that policyholders are more likely to file at least one claim in 2007 and 2010 than in other years. Overall, it is difficult to make conclusions on the overall serial dynamic of the conditional distributions over time solely based on the preliminary analysis (e.g., Tables \ref{tab:data_summary} and \ref{tab:reg_coef}), and hence it is essential to employ the test statistics developed in Section \ref{sec:method} for quantitative assessments.

	\begin{table}[!h]
		\centering
			\caption{[Wisconsin LGPIF dataset] Summary of the logistic regression coefficients and their standard errors across different years. The bolded values represent significance at 5\% level.}
		\resizebox{\textwidth}{!}{%
			\begin{tabular}{lrrrrrrrrrr}
				\toprule
				Year & \multicolumn{2}{c}{2006} & \multicolumn{2}{c}{2007} & \multicolumn{2}{c}{2008} & \multicolumn{2}{c}{2009} & \multicolumn{2}{c}{2010} \\
				Variable & \multicolumn{1}{c}{Estimate} & \multicolumn{1}{c}{SE} & \multicolumn{1}{c}{Estimate} & \multicolumn{1}{c}{SE} & \multicolumn{1}{c}{Estimate} & \multicolumn{1}{c}{SE} & \multicolumn{1}{c}{Estimate} & \multicolumn{1}{c}{SE} & \multicolumn{1}{c}{Estimate} & \multicolumn{1}{c}{SE} \\ \hline
				\texttt{Intercept} & -0.4590 & (0.5837) & 0.6224 & (0.5008) & -0.1773 & (0.5647) & -0.3607 & (0.5287) & 0.2350 & (0.5137) \\
				\texttt{TypeCity} & \textbf{0.8641} & (0.4107) & \textbf{0.8279} & (0.3415) & \textbf{1.3752} & (0.3999) & 0.4529 & (0.3567) & \textbf{0.9746} & (0.3406) \\
				\texttt{TypeCounty} & \textbf{1.1600} & (0.4883) & \textbf{1.2068} & (0.4385) & \textbf{2.2752} & (0.5210) & \textbf{1.5322} & (0.4465) & \textbf{1.5549} & (0.4540) \\
				\texttt{TypeSchool} & 0.3758 & (0.3998) & 0.2478 & (0.3292) & -0.1513 & (0.3988) & -0.4703 & (0.3533) & -0.0534 & (0.3283) \\
				\texttt{TypeTown} & 0.9200 & (0.4728) & -0.2726 & (0.4012) & 0.4330 & (0.4907) & -0.2017 & (0.4417) & 0.3459 & (0.3765) \\
				\texttt{TypeVillage} & \textbf{1.1419} & (0.3961) & \textbf{0.6503} & (0.3078) & \textbf{1.1812} & (0.3799) & \textbf{0.5539} & (0.3283) & \textbf{0.8906} & (0.3071) \\
				\texttt{IsRC} & -0.0363 & (0.1824) & -0.0329 & (0.1739) & 0.0331 & (0.1890) & 0.2112 & (0.1821) & 0.2796 & (0.1728) \\
				\texttt{log(1+CoverageBC)} & \textbf{1.0757} & (0.1037) & \textbf{0.7711} & (0.0908) & \textbf{1.0185} & (0.1051) & \textbf{0.8766} & (0.0977) & \textbf{0.9302} & (0.0947) \\
				\texttt{lnDeductBC} & \textbf{-0.5876} & (0.0789) & \textbf{-0.5441} & (0.0728) & \textbf{-0.6064} & (0.0800) & \textbf{-0.4638} & (0.0765) & \textbf{-0.5404} & (0.0755) \\ 
				\hhline{===========}
			\end{tabular}
		}
		\label{tab:reg_coef}
	\end{table}
	
To assess the serial dynamics of the conditional distribution of $Z_{i,t}$ given $\bm{X}_{i,t}$ over time, we test the null hypothesis of (\ref{null1}) by first computing the pairwise serial dynamic test statistic $\Delta_{s,t,n}$ in (\ref{eq:stat:serial_pair}) for each pair of $(s,t)$. The top right triangle of Table \ref{tab:wald_matrix} presents $\Delta_{s,t,n}$ across different pairs of years, and the bottom left triangle provides the resulting p-values. We find that (\ref{null1}) is rejected under many pairs of $(s,t)$, meaning that the systematic change of conditional distributions over time is prevalent. Looking closely at the p-values, we find that the estimated parameters between years 2008 and 2009 and between years 2007 and 2010 do not differ significantly, but they change significantly from the year 2007 (or 2010) to the year 2008 (or 2009). The aggregate serial dynamic test statistic in (\ref{eq:stat:serial_agg}) becomes $\Delta_{n}=83.1105$ and the corresponding p-value is computed as $p=1.3680\times 10^{-5}\ll 0.05$. As expected, (\ref{null1}) is strongly rejected, suggesting that there are serial dynamics on the conditional claim distributions. 
	
A probable intuition or explanation of the above results is that external environmental states (e.g., social-economic conditions, climate change, and government policy), which may switch over time, influence the claiming behavior of all policyholders simultaneously, causing a heterogeneity of claim distributions over time. For example, the environmental state in the years 2007 and 2010 (2008 and 2009) is classified as ``high-risk" (``low-risk"), so policyholders are generally riskier (less risky).

	\begin{table}[!h]
		\centering
		\small
		\caption{[Wisconsin LGPIF dataset] Pairwise serial dynamic test statistics (top right triangle) and the corresponding p-values (bottom left triangle). The bolded values represent significance at 5\% level.}\label{tab:wald_matrix}
		\begin{subtable}[t]{\textwidth}
			\centering
				\begin{tabular}{l|rrrrr}
				\multicolumn{6}{c}{Serial dynamic test}\\
					\toprule
					Year & \multicolumn{1}{c}{2006} & \multicolumn{1}{c}{2007} & \multicolumn{1}{c}{2008} & \multicolumn{1}{c}{2009} & \multicolumn{1}{c}{2010} \\ \hline
					2006 &  & 13.89 & 21.53 & 18.80 & 29.77 \\
					2007 & 0.1262 &  & 31.75 & 25.41 & 14.61 \\
					2008 & \textbf{0.0105} & \textbf{0.0002} &  & 8.89 & 35.12 \\
					2009 & \textbf{0.0269} & \textbf{0.0026} & 0.4477 &  & 36.16 \\
					2010 & \textbf{0.0005} & 0.1021 & \textbf{0.0001} & \textbf{0.0000} &  \\ 
					\hhline{======}
				\end{tabular}
		\end{subtable}
	\end{table}

To analyze the dependence of $Z_{i,t}$ across different $t$ conditioned on the observed information $(\bm{X}_{i,1},\ldots,\bm{X}_{i,T})$, we first compute the empirical residual correlations $\hat{\rho}_{s,t}$ in (\ref{eq:corr_pair}) for each pair of $(s,t)$ and present them in the left panel of Table \ref{tab:corr_matrix} (bottom left triangle). For comparison, we also provide the sample correlations of $\{(Z_{i,t}, Z_{i,s})\}_{i\in A_{t,s}}$ for each pair of $(s,t)$, which do not control for the covariates, in the top right triangle of the same table. The empirical residual correlations are substantially smaller than the sample correlations for all $(s,t)$, showing that the covariates partially explain the serial dependence of claims. In particular, while the sample correlation remains quite large even with a large time lag (i.e., $l=|s-t|=4$), the empirical residual correlation generally diminishes as $l$ increases. We then test the independence hypotheses of (\ref{null_indep}) and (\ref{null_indep_pair}) by computing the aggregate and pairwise correlation test statistics, $\Gamma_n$ and $\Gamma_{s,t,n}$, in (\ref{eq:test_stat_corr}) and (\ref{eq:test_stat_corr_pair}). The right panel of Table \ref{tab:corr_matrix} presents $\Gamma_{s,t,n}$ accompanied by the corresponding p-values for each pair of $(s,t)$. We see that the claims between year $s$ and $t$ are significantly correlated only when the time lag is no more than two years. Hence, only the most recent claim experiences have significant predictive powers for future claims. As described in Section \ref{sec:problem}, this result favors using dependence models that reflect a decaying or diminishing correlation structure over an increasing time lag while disfavoring models with long-term dependence structures such as symmetric copula and static random effects models. 
Such an empirical finding also coincides with, e.g., \cite{ahn2021ordering}, which shows that the so-called dynamic random effects model following an AR(1) dependence structure outperforms the static random effects model.
Overall, the aggregated correlation test statistic is $\Gamma_n=45.3862$ with a p-value of $p=1.8522\times 10^{-6}\ll 0.05$, strongly rejecting (\ref{null_indep}) and recommending the use of longitudinal models that capture serial dependence.
	\begin{table}[!h]
		\caption{[Wisconsin LGPIF dataset] Left panel: Sample correlations without controlling for covariates (top right triangle) versus empirical residual correlations (bottom left triangle). Right panel: Pairwise independence test statistics (top right triangle) and the corresponding p-values (bottom left triangle).}
		\begin{subtable}[t]{0.49\textwidth}
			\centering
			\resizebox{\textwidth}{!}{%
				\begin{tabular}{c|rrrrr}
				\multicolumn{6}{c}{Correlation values}\\
					\toprule
					Year & \multicolumn{1}{c}{2006} & \multicolumn{1}{c}{2007} & \multicolumn{1}{c}{2008} & \multicolumn{1}{c}{2009} & \multicolumn{1}{c}{2010} \\ \hline
					2006 &  & 0.3023 & 0.2853 & 0.3037 & 0.2853 \\
					2007 & 0.0911 &  & 0.3918 & 0.3021 & 0.2818 \\
					2008 & 0.0127 & 0.1559 &  & 0.3271 & 0.3700 \\
					2009 & 0.0485 & 0.1017 & 0.0793 &  & 0.3142 \\
					2010 & 0.0571 & 0.0600 & 0.1055 & 0.0790 & \\
					\hhline{======}
				\end{tabular}
			}
		\end{subtable}
		\begin{subtable}[t]{0.51\textwidth}
			\centering
			\resizebox{\textwidth}{!}{%
				\begin{tabular}{c|rrrrr}
				\multicolumn{6}{c}{Correlation test}\\
					\toprule
					Year & \multicolumn{1}{c}{2006} & \multicolumn{1}{c}{2007} & \multicolumn{1}{c}{2008} & \multicolumn{1}{c}{2009} & \multicolumn{1}{c}{2010} \\ \hline
					2006 &  & 5.69 & 0.22 & 2.72 & 3.01 \\
					2007 & \textbf{0.0170} &  & 18.48 & 7.09 & 3.36 \\
					2008 & 0.6395 & \textbf{0.0000} &  & 4.05 & 9.27 \\
					2009 & 0.0993 & \textbf{0.0077} & \textbf{0.0441} &  & 5.60 \\
					2010 & 0.0826 & 0.0666 & \textbf{0.0023} & \textbf{0.0179} &  \\
					\hhline{======}
				\end{tabular}
			}
		\end{subtable}
		\label{tab:corr_matrix}
	\end{table}

\subsection{French private motor dataset}
We analyze the French private motor dataset, which is publicly available in the \texttt{R} package called \texttt{CASdatasets} by retrieving \texttt{data(fremotor2freq9907b)}. The dataset records the annual claim frequencies among $n=72,479$ policyholders from 1999 to 2007. For illustrative purposes, we only analyze the years 2001 to 2006, the period when the number of observations $n_t$ is sufficiently large and is rapidly changing over time. Explanatory variables accompanying each policyholder are described in Table \ref{tab:prelim_x_french}. Note that the three categorical variables, vehicle usage (\texttt{Usage}), vehicle type (\texttt{VehType}), and vehicle power (\texttt{VehPower}), originally have 18, 15, and 8 levels, respectively, and the distribution of these variables are severely imbalanced. 
As a result, we aggregate these levels (see Table \ref{tab:prelim_x_french}) to ensure a sufficient number of observations for each level of each variable. 
Investigation of the optimal way to aggregate the variable levels is out of the scope of this paper. 

\begin{table}[!h]
\caption{\label{tab:prelim_x_french}[French motor dataset] Descriptive summary for the explanatory variables.}
\centering
\resizebox{\textwidth}{!}{%
\begin{tabular}{lllll}
\toprule
\multicolumn{1}{c}{Index} & \multicolumn{1}{c}{Variable name} & \multicolumn{1}{c}{Type} & \multicolumn{1}{c}{Levels} & \multicolumn{1}{c}{Description} \\
\midrule
1 & \texttt{Expo} & Continuous & -- & Policyholder exposure over a particular year (between 0 and 1). \\
2 & \texttt{NEW} & Binary & -- & Indicator of a new insurance contract. \\
3--9 & \texttt{Usage} & Categorical & A--H & Vehicle Usage. Levels A to G map to Usage 10, 11, 14, 15, 16, 18, and   5 \\
 &  &  &  & from the original data. Level H is a reference category for other usages. \\
10--16 & \texttt{VehType} & Categorical & A--H & Vehicle Type. Levels A to G map to Type 10, 11, 3, 6, 7, 8, and 9 from the \\
 &  &  &  & original data. Level H is a reference category for other types. \\
17--20 & \texttt{VehPower} & Categorical & A--E & Vehicle Power. Levels A to D map to Power levels 1 to 4 from the original \\
 &  &  &  & data. Level E is a reference category for more powerful vehicles.\\
\hhline{=====}
\end{tabular}
}
\end{table}

	\begin{table}[!h]
		\small
	\caption{[French motor dataset] Summary statistics of the observations across different years.}
		\centering
			\begin{tabular}{lrrrrrr}
				\toprule
				Year & \multicolumn{1}{c}{2001} & \multicolumn{1}{c}{2002} & \multicolumn{1}{c}{2003} & \multicolumn{1}{c}{2004} & \multicolumn{1}{c}{2015} & \multicolumn{1}{c}{2016}\\ \hline
				Number of observations & 22491 & 31035 & 40742 & 50450 & 60957 & 72749 \\
				Proportion of non-zero claims & 0.1329 & 0.1278 & 0.1291 & 0.1326 & 0.1282 & 0.1245 \\
				\hhline{=======}
			\end{tabular}
		\label{tab:data_summary_french}
	\end{table}
	
Following the procedures of analyzing the LGPIF dataset in Section \ref{sec:data_lgpif}, we first present the summary statistics for the French motor dataset in Table \ref{tab:data_summary_french}. As $n_t$ increases over time, new policyholders join the pool yearly. About 13\% of the policyholders file at least one claim each year, and this number does not fluctuate significantly over time. The estimated logistic regression parameters with the standard errors for each year are then reported in Table \ref{tab:reg_coef_french}. The regression coefficients do not vary substantially over time. 
	
To test the null hypothesis of (\ref{null1}) for no serial dynamics, we present a matrix of the pairwise serial dynamic test statistic $\Delta_{s,t,n}$ (Equation (\ref{eq:stat:serial_pair})) across all pairs of years with the corresponding p-values in the left panel of Table \ref{tab:wald_matrix_french} (full model). Since all the p-values are above the significance level of 0.05, there is no apparent structural difference in conditional claim distributions between years $s$ and $t$. Moreover, the aggregate serial dynamic test statistic (Equation (\ref{eq:stat:serial_agg})) is $\Delta_{n}=114.8076$ with a p-value $p=0.2412>0.05$, showing that the null hypothesis of (\ref{null1}) is not rejected.

	\begin{table}[!h]
		\centering
		\caption{[French motor dataset] Summary of the logistic regression coefficients and their standard errors. The bolded values represent significance at 5\% level.}
		\resizebox{\textwidth}{!}{%
			\begin{tabular}{lrrrrrrrrrrrr}
				\toprule
				Year & \multicolumn{2}{c}{2001} & \multicolumn{2}{c}{2002} & \multicolumn{2}{c}{2003} & \multicolumn{2}{c}{2004} & \multicolumn{2}{c}{2005} & \multicolumn{2}{c}{2006} \\
				Variable & \multicolumn{1}{c}{Estimate} & \multicolumn{1}{c}{SE} & \multicolumn{1}{c}{Estimate} & \multicolumn{1}{c}{SE} & \multicolumn{1}{c}{Estimate} & \multicolumn{1}{c}{SE} & \multicolumn{1}{c}{Estimate} & \multicolumn{1}{c}{SE} & \multicolumn{1}{c}{Estimate} & \multicolumn{1}{c}{SE} & \multicolumn{1}{c}{Estimate} & \multicolumn{1}{c}{SE} \\ \hline
				\texttt{Intercept} & \textbf{-2.6279} & (0.2983) & \textbf{-2.2934} & (0.2467) & \textbf{-2.6516} & (0.2163) & \textbf{-2.6749} & (0.2053) & \textbf{-2.4768} & (0.1856) & \textbf{-2.9873} & (0.1497) \\
				\texttt{Expo} & \textbf{1.6829} & (0.1707) & \textbf{1.6577} & (0.1553) & \textbf{1.7896} & (0.1414) & \textbf{1.8881} & (0.1505) & \textbf{1.5834} & (0.1366) & \textbf{2.0139} & (0.0941) \\
				\texttt{NEW} & \textbf{0.1435} & (0.0481) & 0.0764 & (0.0458) & \textbf{0.2100} & (0.0421) & \textbf{0.1081} & (0.0418) & \textbf{0.1191} & (0.0404) & \textbf{0.1661} & (0.0366) \\
				\texttt{UsageA} & \textbf{-0.5973} & (0.2161) & \textbf{-0.5447} & (0.1718) & \textbf{-0.4693} & (0.1418) & \textbf{-0.2855} & (0.1228) & \textbf{-0.2802} & (0.1100) & \textbf{-0.3787} & (0.1026) \\
				\texttt{UsageB} & -0.2298 & (0.1424) & -0.2010 & (0.1138) & \textbf{-0.2713} & (0.0973) & -0.0354 & (0.0852) & -0.0198 & (0.0766) & -0.1362 & (0.0704) \\
				\texttt{UsageC} & -0.3555 & (0.2610) & \textbf{-0.9650} & (0.2412) & \textbf{-0.4636} & (0.1547) & \textbf{-0.5592} & (0.1418) & \textbf{-0.4089} & (0.1188) & \textbf{-0.5231} & (0.1083) \\
				\texttt{UsageD} & 0.1983 & (0.1103) & 0.0646 & (0.0861) & 0.0853 & (0.0700) & \textbf{0.2155} & (0.0620) & \textbf{0.2166} & (0.0552) & \textbf{0.1699} & (0.0500) \\
				\texttt{UsageE} & \textbf{-1.1533} & (0.2071) & \textbf{-1.2651} & (0.1717) & \textbf{-1.0950} & (0.1346) & \textbf{-0.7246} & (0.1107) & \textbf{-1.0174} & (0.1086) & \textbf{-0.9716} & (0.0959) \\
				\texttt{UsageF} & 0.0132 & (0.1398) & -0.2026 & (0.1130) & -0.1516 & (0.0937) & -0.1234 & (0.0839) & -0.0768 & (0.0750) & \textbf{-0.1728} & (0.0693) \\
				\texttt{UsageG} & \textbf{-0.5374} & (0.1184) & \textbf{-0.6540} & (0.0935) & \textbf{-0.6396} & (0.0768) & \textbf{-0.4880} & (0.0677) & \textbf{-0.5108} & (0.0607) & \textbf{-0.5772} & (0.0554) \\
				\texttt{VehTypeA} & \textbf{0.6047} & (0.1960) & \textbf{0.3997} & (0.1547) & \textbf{0.5421} & (0.1365) & \textbf{0.5027} & (0.1169) & \textbf{0.5229} & (0.1071) & \textbf{0.5603} & (0.1009) \\
				\texttt{VehTypeB} & \textbf{0.5246} & (0.2251) & \textbf{0.6066} & (0.1776) & \textbf{0.6537} & (0.1551) & \textbf{0.6487} & (0.1359) & \textbf{0.5804} & (0.1238) & \textbf{0.6118} & (0.1162) \\
				\texttt{VehTypeC} & 0.3714 & (0.2284) & 0.2813 & (0.1808) & \textbf{0.4163} & (0.1578) & \textbf{0.4104} & (0.1355) & \textbf{0.5930} & (0.1224) & \textbf{0.4716} & (0.1151) \\
				\texttt{VehTypeD} & \textbf{-0.4348} & (0.2200) & \textbf{-0.7136} & (0.1767) & \textbf{-0.4599} & (0.1531) & \textbf{-0.5931} & (0.1320) & \textbf{-0.5552} & (0.1209) & \textbf{-0.4553} & (0.1128) \\
				\texttt{VehTypeE} & -0.3494 & (0.2627) & \textbf{-0.6352} & (0.2195) & \textbf{-0.4656} & (0.1880) & \textbf{-0.5587} & (0.1661) & \textbf{-0.3830} & (0.1469) & \textbf{-0.4163} & (0.1387) \\
				\texttt{VehTypeF} & \textbf{-0.8022} & (0.3754) & -0.2881 & (0.2483) & -0.1820 & (0.2079) & 0.0286 & (0.1733) & -0.0041 & (0.1591) & 0.0810 & (0.1458) \\
				\texttt{VehTypeG} & \textbf{-2.5736} & (0.3537) & \textbf{-2.7288} & (0.2777) & \textbf{-2.6548} & (0.2389) & \textbf{-3.0436} & (0.2159) & \textbf{-2.7649} & (0.1830) & \textbf{-2.6723} & (0.1696) \\
				\texttt{VehPowerA} & \textbf{-2.8699} & (0.1455) & \textbf{-2.9151} & (0.1203) & \textbf{-2.7994} & (0.0994) & \textbf{-3.0308} & (0.0878) & \textbf{-2.9733} & (0.0799) & \textbf{-2.8638} & (0.0732) \\
				\texttt{VehPowerB} & \textbf{-1.2881} & (0.1343) & \textbf{-1.3722} & (0.1095) & \textbf{-1.3312} & (0.0901) & \textbf{-1.4734} & (0.0779) & \textbf{-1.3998} & (0.0703) & \textbf{-1.3243} & (0.0639) \\
				\texttt{VehPowerC} & \textbf{-0.6938} & (0.1310) & \textbf{-0.6956} & (0.1063) & \textbf{-0.6464} & (0.0870) & \textbf{-0.7588} & (0.0749) & \textbf{-0.7495} & (0.0675) & \textbf{-0.6366} & (0.0611) \\
				\texttt{VehPowerD} & \textbf{-0.6503} & (0.1335) & \textbf{-0.5907} & (0.1081) & \textbf{-0.5636} & (0.0885) & \textbf{-0.5643} & (0.0759) & \textbf{-0.5611} & (0.0686) & \textbf{-0.4390} & (0.0620) \\ 
				\hhline{=============}
			\end{tabular}
		}
		\label{tab:reg_coef_french}
	\end{table}
	
\begin{table}[!h]
\caption{[French motor dataset] Pairwise serial dynamic test statistics (top right triangles) and the corresponding p-values (bottom left triangles) under the full model (left panel) and the reduced model (right panel).}
\begin{subtable}[t]{0.49\textwidth}
\centering
\resizebox{\textwidth}{!}{%
\begin{tabular}{l|rrrrrr}
\multicolumn{7}{c}{Serial dynamic test (Full model)} \\
\toprule
Year & \multicolumn{1}{c}{2001} & \multicolumn{1}{c}{2002} & \multicolumn{1}{c}{2003} & \multicolumn{1}{c}{2004} & \multicolumn{1}{c}{2005} & \multicolumn{1}{c}{2006} \\ \hline
2001 &  & 18.12 & 14.08 & 30.32 & 26.52 & 31.46 \\
2002 & 0.6411 &  & 13.22 & 21.94 & 14.12 & 20.26 \\
2003 & 0.8661 & 0.9006 &  & 27.07 & 20.66 & 18.74 \\
2004 & 0.0858 & 0.4029 & 0.1685 &  & 21.09 & 18.53 \\
2005 & 0.1874 & 0.8645 & 0.4798 & 0.4537 &  & 17.43 \\
2006 & 0.0663 & 0.5051 & 0.6015 & 0.6150 & 0.6847 & \\ 
\hhline{=======}
\end{tabular}
}
\end{subtable}
\begin{subtable}[t]{0.49\textwidth}
\centering
\resizebox{\textwidth}{!}{%
\begin{tabular}{c|rrrrrr}
\multicolumn{7}{c}{Serial dynamic test (Reduced model)} \\
\toprule
Year & \multicolumn{1}{c}{2001} & \multicolumn{1}{c}{2002} & \multicolumn{1}{c}{2003} & \multicolumn{1}{c}{2004} & \multicolumn{1}{c}{2005} & \multicolumn{1}{c}{2006} \\ \hline
2001 &  & 10.11 & 11.06 & 24.12 & 26.24 & 34.99 \\
2002 & 0.6849 &  & 3.46 & 14.76 & 8.06 & 14.43 \\
2003 & 0.6056 & 0.9957 &  & 14.96 & 12.63 & 18.28 \\
2004 & \textbf{0.0300} & 0.3227 & 0.3102 &  & 18.95 & 18.93 \\
2005 & \textbf{0.0158} & 0.8400 & 0.4771 & 0.1246 &  & 15.86 \\
2006 & \textbf{0.0008} & 0.3443 & 0.1473 & 0.1252 & 0.2566 &  \\
\hhline{=======}
\end{tabular}
}
\end{subtable}
\label{tab:wald_matrix_french}
\end{table}
	
We then investigate how the selection of covariates can influence the diagnostic test results. For the expository purpose, we further consider a ``reduced model" where we remove the variables in Table \ref{tab:reg_coef_french} that are insignificant in any of the six years. For example, the regression coefficients of levels B, C, D, and F of the \texttt{Usage} variable are insignificant in at least one of six years, so we merge these levels into the reference category of the variable \texttt{Usage}. In this case, we resemble an insurance company with a relatively poor underwriting process, failing to collect some policyholder information for pricing. After that, we perform the same procedures to compute the pairwise and aggregate serial dynamic test statistics $\Delta_{s,t,n}$ and $\Delta_n$ with a reduced number of variables included. The pairwise test results are presented in the right panel of Table \ref{tab:wald_matrix_french}, while the aggregate test statistic is $\Delta_n=95.4390$ (p-value $0.0083$). In contrast to the diagnostic results for the entire model, the null hypothesis of (\ref{null1}) is rejected under the reduced model. This suggests that there is a serial dynamic structure on the claims $(Z_{i,1},\ldots,Z_{i,T})$ over time conditioned on less information. 
Therefore, a complete underwriting process to collect as much useful policyholder information as possible is essential to ensure that the serial dynamics of a claim process are fully explained by the covariates, making predictive modeling of actuarial data more feasible even if $T$ is small.

Returning to the full model, we now test the conditional independence null hypotheses of (\ref{null_indep}) and (\ref{null_indep_pair}) by performing the correlation test for the French motor dataset. The left panel of Table \ref{tab:corr_french} showcases the sample correlations of $\{(Z_{i,t},Z_{i,s})\}_{i\in A_{t,s}}$ and the empirical residual correlations $\hat{\rho}_{s,t}$. Similar to the LGPIF dataset, the correlations from the French dataset are reduced after controlling for the covariates. However, the residual correlations for the French dataset do not decay as much as the LGPIF dataset as the time lag $l$ increases. In other words, the older claims may still be useful for predicting future claims. The right panel of Table \ref{tab:corr_french} presents the pairwise correlation test statistic $\Gamma_{s,t,n}$ and the p-values for each pair of $(s,t)$. Unlike the LGPIF dataset where $\Gamma_{s,t,n}$ is only significant when $l=|s-t|\leq 2$, we see that regardless of $s$ and $t$, $Z_{i,t}$ and $Z_{i,s}$ are significantly dependent conditioned on $(\bm{X}_{i,t},\bm{X}_{i,s})$ for the French data. This suggests using a long-term memory model than a short-term memory model in capturing the serial dependence of the French motor claims. Unsurprisingly, the aggregated correlation test statistic ($\Gamma_n=1073.04$ with $p<10^{-16}$) strongly rejects the null hypothesis of (\ref{null_indep}) as well. 

	\begin{table}[!h]
	\caption{[French motor dataset] Left panel: Sample (top right triangle) versus empirical residual correlations (bottom left triangle). Right panel: Pairwise independence test statistics (top right triangle) with p-values (bottom left triangle).}
		\begin{subtable}[t]{0.49\textwidth}
			\centering
			\resizebox{\textwidth}{!}{%
				\begin{tabular}{c|rrrrrr}
				\multicolumn{7}{c}{Correlation values} \\
					\toprule
					Year & \multicolumn{1}{c}{2001} & \multicolumn{1}{c}{2002} & \multicolumn{1}{c}{2003} & \multicolumn{1}{c}{2004} & \multicolumn{1}{c}{2005} & \multicolumn{1}{c}{2006} \\ \hline
					2001 &  & 0.2345 & 0.2243 & 0.2114 & 0.2206 & 0.2131 \\
					2002 & 0.1424 &  & 0.2371 & 0.2499 & 0.2309 & 0.2189 \\
					2003 & 0.1285 & 0.1327 &  & 0.2474 & 0.2364 & 0.2248 \\
					2004 & 0.1127 & 0.1375 & 0.1453 &  & 0.2361 & 0.2388 \\
					2005 & 0.1110 & 0.1305 & 0.1305 & 0.1293 &  & 0.2346 \\
					2006 & 0.1068 & 0.1127 & 0.1372 & 0.1353 & 0.1186 & \\
					\hhline{=======}
				\end{tabular}
			}
		\end{subtable}
		\begin{subtable}[t]{0.53\textwidth}
			\centering
			\resizebox{\textwidth}{!}{%
				\begin{tabular}{c|rrrrrr}
				\multicolumn{7}{c}{Correlation test}\\
					\toprule
					Year & \multicolumn{1}{c}{2001} & \multicolumn{1}{c}{2002} & \multicolumn{1}{c}{2003} & \multicolumn{1}{c}{2004} & \multicolumn{1}{c}{2005} & \multicolumn{1}{c}{2006} \\ \hline
					2001 &  & 141.28 & 134.61 & 130.76 & 144.33 & 131.53 \\
					2002 & \textbf{0.0000} &  & 200.49 & 206.07 & 167.69 & 155.91 \\
					2003 & \textbf{0.0000} & \textbf{0.0000} &  & 241.96 & 265.00 & 192.91 \\
					2004 & \textbf{0.0000} & \textbf{0.0000} & \textbf{0.0000} &  & 277.62 & 250.62 \\
					2005 & \textbf{0.0000} & \textbf{0.0000} & \textbf{0.0000} & \textbf{0.0000} &  & 329.02 \\
					2006 & \textbf{0.0000} & \textbf{0.0000} & \textbf{0.0000} & \textbf{0.0000} & \textbf{0.0000} &  \\
					\hhline{=======}
				\end{tabular}
			}
		\end{subtable}
		\label{tab:corr_french}
	\end{table}

\section{Simulation study}\label{sec:simu}
This section investigates the finite sample performance of the proposed tests and compares our tests to some standard tests in the literature. For simplicity, we study $T=2$ years only and consider the following two cases for the sample size: (i) $n_t=1117$ for both years, which is the number of observations in the LGPIF dataset under the BC peril in the year 2009; (ii) a larger sample size of $n_t=4468=1117\times 4$. For case (i), the covariates $\bm{X}_{i,t}$ of the $i$-th observation in year $t$ are directly copied from the covariates of the $i$-th observation of the LGPIF dataset in the year 2009, for $i=1,\ldots,n_t$ and $t=1,2$. For case (ii), $\{\bm{X}_{i,t}\}_{i=1,\ldots,n_t}$ is a fourfold duplicate from the LGPIF dataset (the year 2009). This simplified data structure is balanced with $n=n_1=n_2$. Also, note that $\bm{X}_{i,1}=\bm{X}_{i,2}$ for $i=1,\ldots,n$. The simulation study comprises the following steps:
	
\begin{itemize}
\item Step \romannumeral3 1) For $i=1,\ldots,n_t$, $Z_{i,1}$ given $\bm{X}_{i,1}$ is independently simulated from a logistic regression model with the regression coefficients given by the estimated values from the real dataset in the year 2009, see Table \ref{tab:reg_coef}. With an independent probability of $q\in[0,1]$, we set $Z_{i,2}=Z_{i,1}$. Otherwise, with a probability of $1-q$, $Z_{i,2}$ given $\bm{X}_{i,2}$ is independently simulated from a logistic regression model with the same set of regression coefficients.
\item Step \romannumeral3 2) For $t=1,2$, the simulated data $(Z_{i,t},\bm{X}_{i,t})_{i=1,\ldots,n_t}$ is fitted to a logistic regression model. The estimated parameters $\hat{\bm{\gamma}}_t^{(r)}$ are obtained.
\item Step \romannumeral3 3) Compute the serial dynamic test statistic $\Delta_n^{(r)}$ in (\ref{eq:stat:serial_agg}) and the correlation test statistic $\Gamma_n^{(r)}$ in (\ref{eq:test_stat_corr}) using the random weighted bootstrap method. Calculate the corresponding p-values, denoted as $\hat{p}_{1,n}^{(r)}$ and $\hat{p}_{2,n}^{(r)}$ for the serial dynamic test and the correlation test respectively.
\item Step iii4) Repeat the above three steps $R$ times to get $\{\hat{p}_{1,n}^{(r)}\}_{r=1,\ldots,R}$ and $\{\hat{p}_{2,n}^{(r)}\}_{r=1,\ldots,R}$.
\end{itemize}
	
It is easily to verify from the first step that $Z_{i,2}$ given $\bm{X}_{i,2}$ also marginally follows a logistic regression model. In this experiment, we also consider two choices of probability $q$: (a) $q=0$ ($Z_{i,1}$ and $Z_{i,2}$ are conditionally independent given $\bm{X}_{i,1}$ and $\bm{X}_{i,2}$); (b) $q=0.5$ ($Z_{i,1}$ and $Z_{i,2}$ are positively related). We employ $B=1000$ and $R=1000$. Theorems \ref{Th3} and \ref{Th6} suggest that both $\{\hat{p}_{1,n}^{(r)}\}_{r=1,\ldots,R}$ and $\{\hat{p}_{2,n}^{(r)}\}_{r=1,\ldots,R}$ approximately follow a $\text{Uniform}[0,1]$ distribution as $n\rightarrow\infty$ when $q=0$. When $q=0.5$, $\{\hat{p}_{1,n}^{(r)}\}_{r=1,\ldots,R}$ should still asymptotically follow a standard uniform distribution, but the correlation test should strongly reject the null.

\subsection{Finite sample performance of proposed tests}
Figures \ref{fig:sim:hist1} and \ref{fig:sim:hist2} present the empirical densities for $\{\hat{p}_{1,n}^{(r)}\}_{r=1,\ldots,R}$ (left panel) and $\{\hat{p}_{2,n}^{(r)}\}_{r=1,\ldots,R}$ (right panel) under $n_t=1117$ and $n_t=4468$ respectively when $q=0$. When $n_t=1117$, the empirical distribution of $\hat{p}_{1,n}^{(r)}$ looks reasonably uniform but the density of $\hat{p}_{2,n}^{(r)}$ is slightly tilted towards smaller p-values. The rejection probabilities at $(1\%, 5\%, 10\%)$ significance levels are $(1.2\%, 4.5\%, 9.7\%)$ for the serial dynamic test and $(2.3\%, 7.7\%, 14.0\%)$ for the correlation test. The rejection probabilities of the serial dynamic test are very close to the desired significance levels. Hence, $n_t=1117$ is a reasonable sample size for the serial dynamic test to perform satisfactorily for our LGPIF dataset. On the other hand, the rejection probabilities of the correlation test are consistently higher than the desired levels, indicating that the correlation test rejects the null hypothesis slightly more often than it should when the sample size is not larger. With a sufficiently large sample size ($n_t=4468$), the empirical distributions of both p-values seem uniform, and the rejection probabilities at $(1\%,5\%,10\%)$ levels are $(1.2\%,5.9\%,12.3\%)$ for the serial dynamic test and $(0.7\%,4.2\%,9.4\%)$ for the correlation test. Therefore, a larger sample size can effectively mitigate the bias of the correlation test p-values. 

Performing a similar analysis, Figure \ref{fig:sim:hist3} exhibits the empirical densities of the serial dynamic test p-value when $q=0.5$. The serial dynamic test still behaves properly when the claim indicators are conditionally dependent. Under this test, the rejection probabilities at $(1\%,5\%,10\%)$ levels are $(0.8\%,4.7\%,10.0\%)$ ($n_t=1117$) and $(1.1\%,4.4\%,10.1\%)$ ($n_t=4468$) respectively. When $q=0.5$, the correlation test always strongly rejects the null under both cases for the sample size with p-values very close to zero. For conciseness, we do not provide the corresponding density plots.

\begin{figure}[!h]
\begin{center}
\begin{subfigure}[h]{0.49\linewidth}
\includegraphics[width=\linewidth]{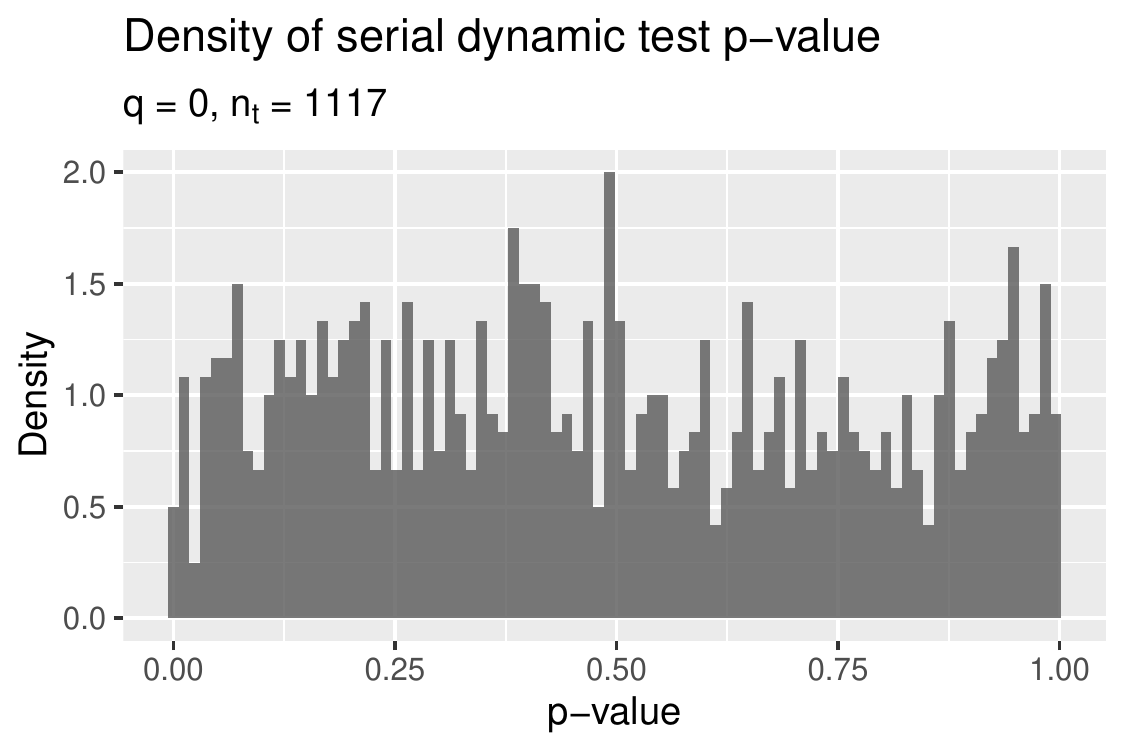}
\end{subfigure}
\begin{subfigure}[h]{0.49\linewidth}
\includegraphics[width=\linewidth]{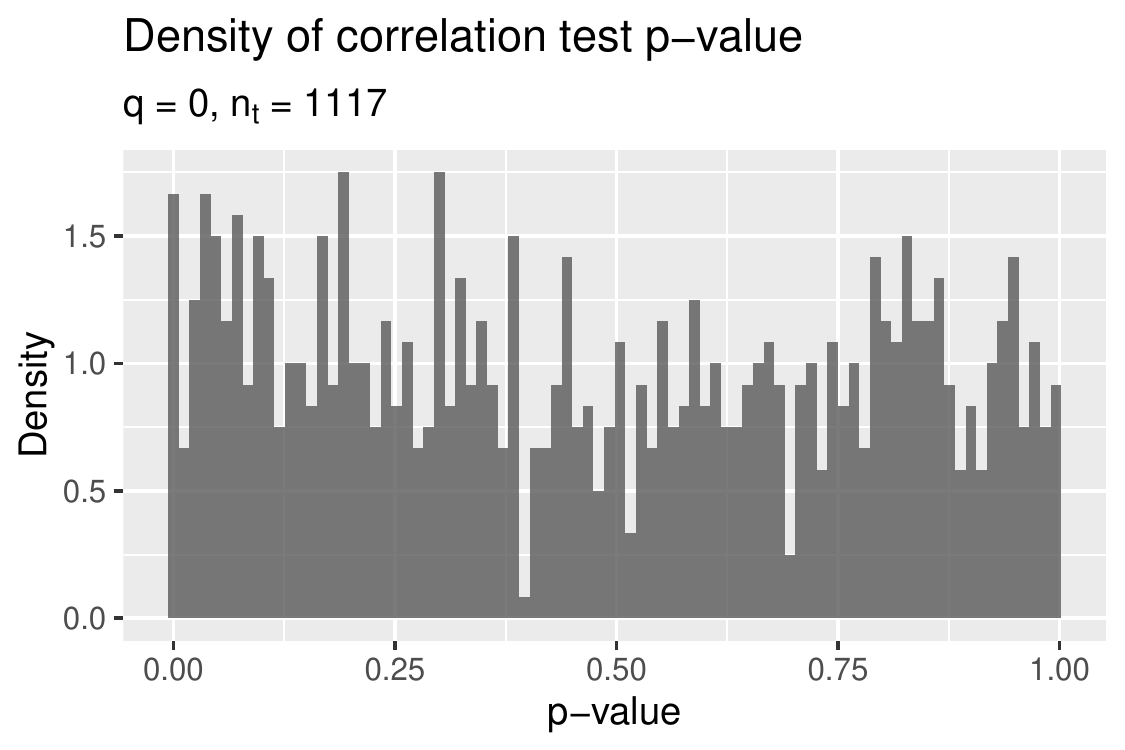}
\end{subfigure}
\end{center}
\vspace{-1.5em}
\caption{[Simulation study] Empirical distributions of the serial dynamic test p-values $\hat{p}_{1,n}^{(r)}$ (left panel) and the correlation test p-values $\hat{p}_{2,n}^{(r)}$ (right panel) under $n_t=1117$ when $q=0$.}
\label{fig:sim:hist1}
\end{figure}

\begin{figure}[!h]
\begin{center}
\begin{subfigure}[h]{0.49\linewidth}
\includegraphics[width=\linewidth]{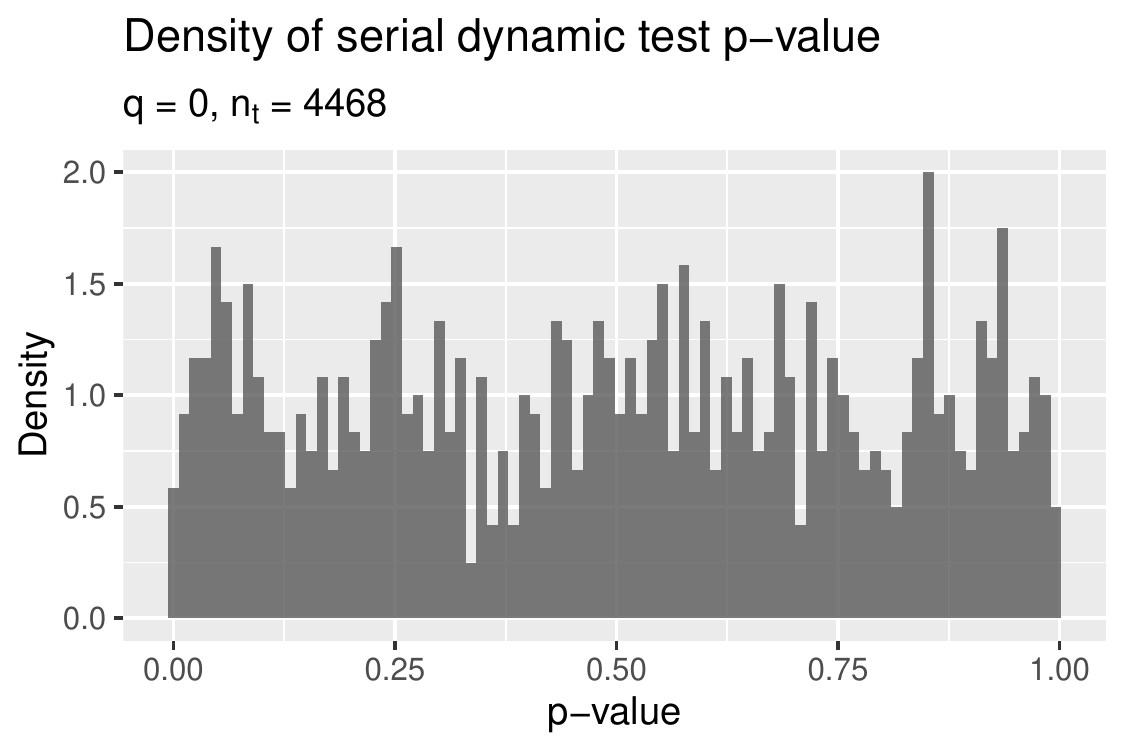}
\end{subfigure}
\begin{subfigure}[h]{0.49\linewidth}
\includegraphics[width=\linewidth]{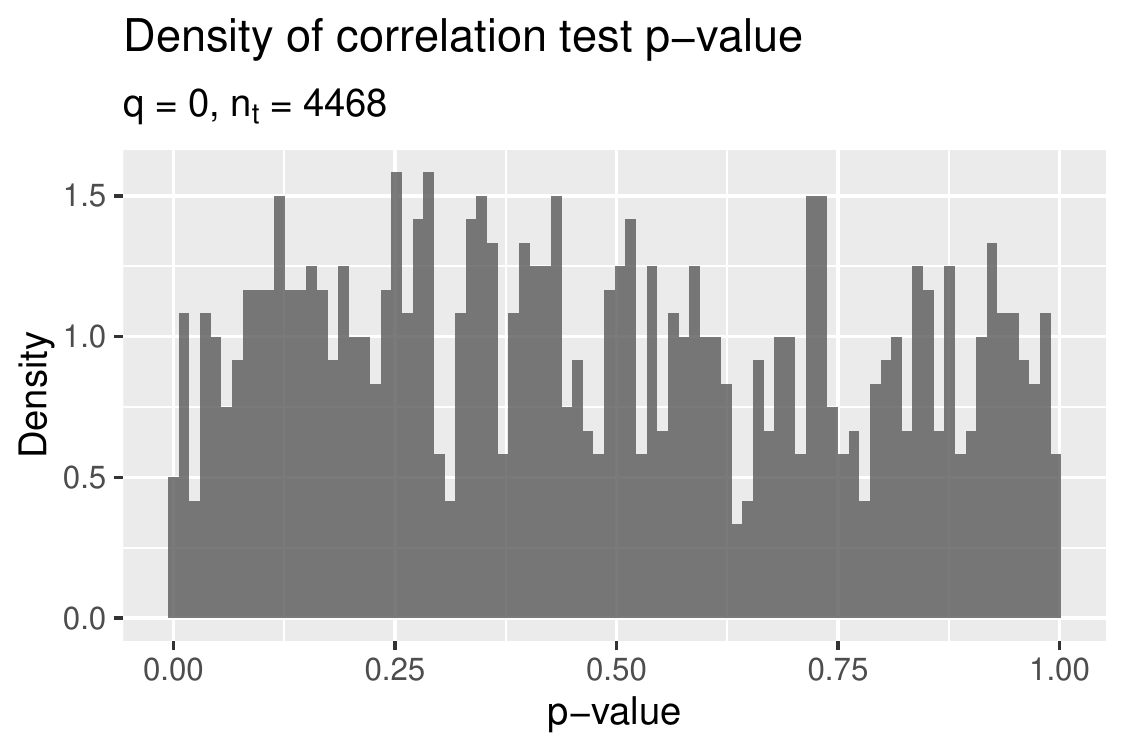}
\end{subfigure}
\end{center}
\vspace{-1.5em}
\caption{[Simulation study] Empirical distributions of the serial dynamic test p-values $\hat{p}_{1,n}^{(r)}$ (left panel) and the correlation test p-values $\hat{p}_{2,n}^{(r)}$ (right panel) under $n_t=4468$ when $q=0$.}
\label{fig:sim:hist2}
\end{figure}

\begin{figure}[!h]
\begin{center}
\begin{subfigure}[h]{0.49\linewidth}
\includegraphics[width=\linewidth]{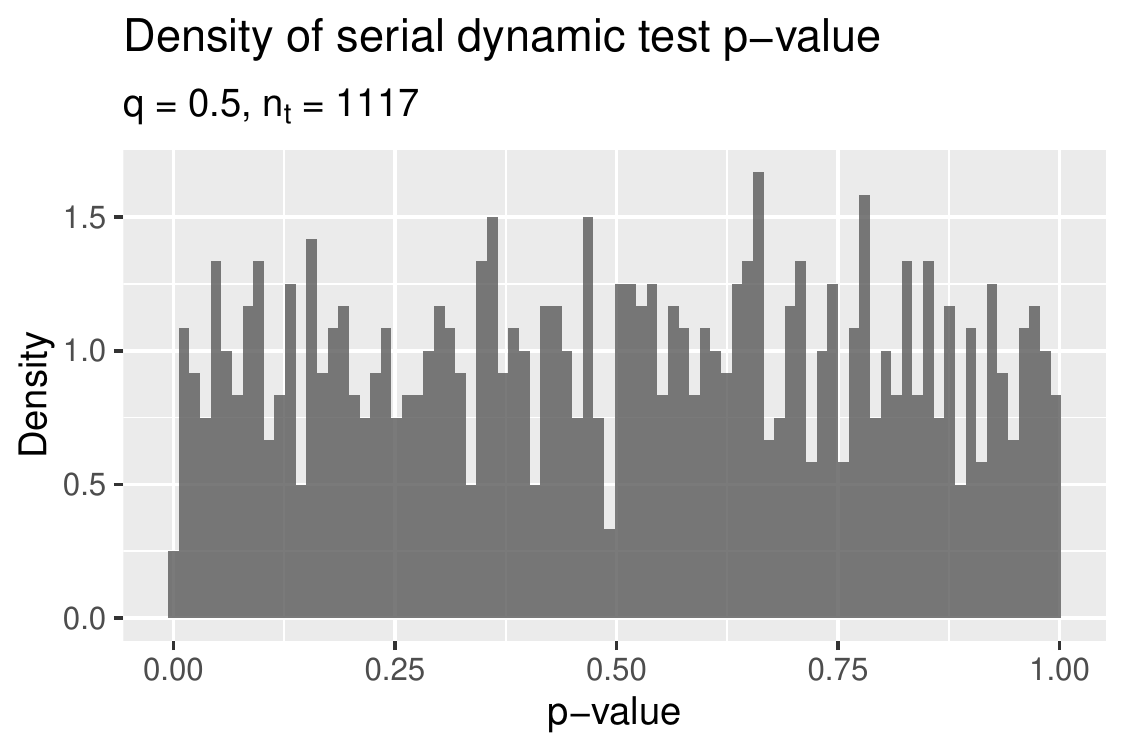}
\end{subfigure}
\begin{subfigure}[h]{0.49\linewidth}
\includegraphics[width=\linewidth]{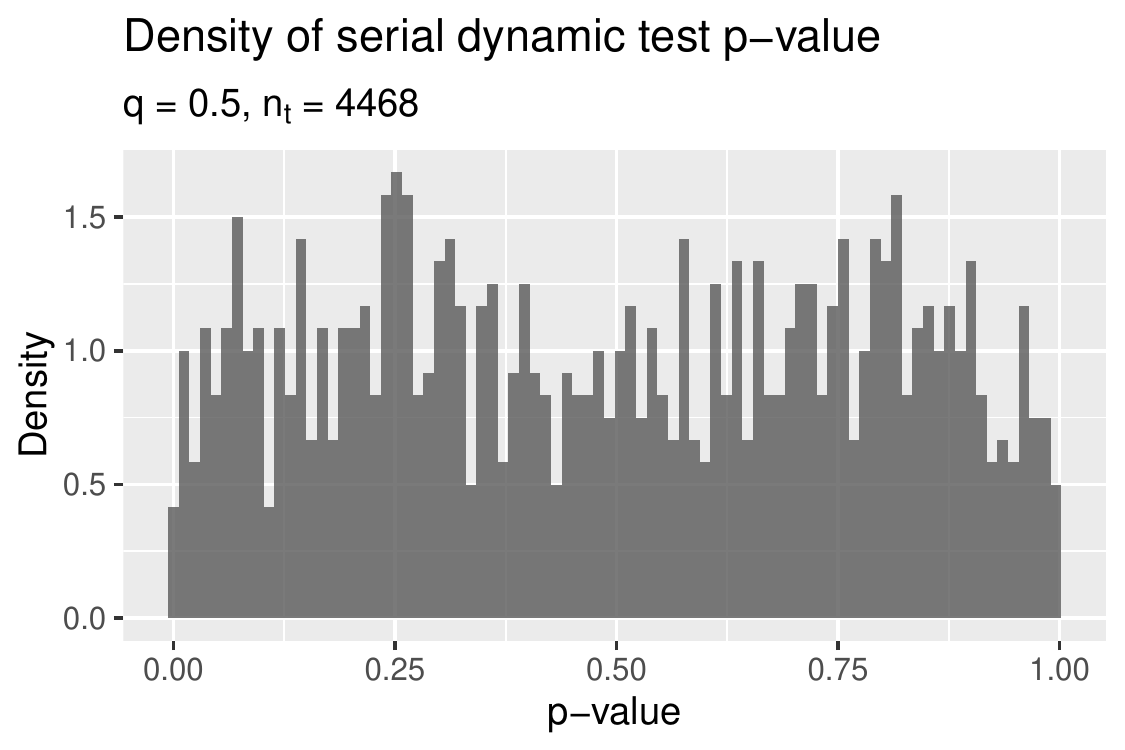}
\end{subfigure}
\end{center}
\vspace{-1.5em}
\caption{[Simulation study] Empirical distributions of serial dynamic test p-values $\hat{p}_{1,n}^{(r)}$ under $n_t=1117$ (left panel) and $n_t=4468$ (right panel) when $q=0.5$.}
\label{fig:sim:hist3}
\end{figure}

\subsection{Comparison studies}
With a simplified data structure (balanced data with $T=2$ years only), we may attempt to apply some standard tests in the literature to this simulation study. We aim to show that our proposed tests outperform the existing methods.
\subsubsection{Serial dynamic test}
To test the null hypothesis of (\ref{null1}), one may alternatively perform a \emph{naive likelihood ratio test} as follows. If we assume that $Z_{i,1}$ given $\bm{X}_{i,1}$ is conditionally independent of $Z_{i,2}$ given $\bm{X}_{i,2}$, the joint log-likelihood is given by
\begin{equation}
\mathcal{L}(\bm{\gamma}_1,\bm{\gamma}_2)=\sum_{t=1}^2\sum_{i=1}^{n_t}\left\{Z_{i,t}\log(p_{i,t})+(1-Z_{i,t})\log(1-p_{i,t})\right\}.
\end{equation}

Then, the MLE under the full model is given by $(\hat{\bm{\gamma}}_1,\hat{\bm{\gamma}}_2)=\arg\max_{(\bm{\gamma}_1,\bm{\gamma}_2)}\mathcal{L}(\bm{\gamma}_1,\bm{\gamma}_2)$ and the MLE under the reduced model is $\tilde{\bm{\gamma}}=\arg\max_{\bm{\gamma}_1}\mathcal{L}(\bm{\gamma}_1,\bm{\gamma}_1)$. The likelihood ratio test statistic for (\ref{null1}) is then constructed as $\Delta^{\text{LR}}_{1,2,n}=-2\times\left[\mathcal{L}(\tilde{\bm{\gamma}},\tilde{\bm{\gamma}})-\mathcal{L}(\hat{\bm{\gamma}}_1,\hat{\bm{\gamma}}_2)\right].$ Standard results show that $\Delta^{\text{LR}}_{1,2,n}$ asymptotically follows $\chi^2(P+1)$, so we reject (\ref{null1}) at level $a$ whenever $\Delta^{\text{LR}}_{1,2,n}>\chi^2_{1-a,P+1}$. Considering two cases of $q\in\{0,0.5\}$, we employ the naive likelihood ratio test to $R=1000$ simulated datasets with a sample size of $n_t=4468$ each, producing $1000$ p-values $\{\tilde{p}_{1,n}^{(r)}\}_{r=1,\ldots,R}$. Figure \ref{fig:sim:hist_naive} presents the empirical distributions of $\tilde{p}_{1,n}^{(r)}$ under the two choices of $q$. When $q=0$, the likelihood ratio test correctly specifies conditional independence assumption, so the empirical distribution of $\tilde{p}_{1,n}^{(r)}$ looks uniform, and the rejection probabilities at $(1\%,5\%,10\%)$ levels are $(0.7\%,5.9\%,11.3\%)$. Meanwhile, when $q=0.5$, the naive likelihood ratio test misspecifies the serial dependence assumption, so the empirical distribution of $\tilde{p}_{1,n}^{(r)}$ are severely distorted with the rejection probabilities at $(1\%,5\%,10\%)$ levels being $(0.0\%,0.0\%,0.1\%)$. In this case, the naive likelihood ratio test does not properly test the null. Comparing the right panel of Figure \ref{fig:sim:hist3} to that of Figure \ref{fig:sim:hist_naive}, the proposed serial dynamic test is more robust than the naive likelihood ratio test.

\begin{figure}[!h]
\begin{center}
\begin{subfigure}[h]{0.49\linewidth}
\includegraphics[width=\linewidth]{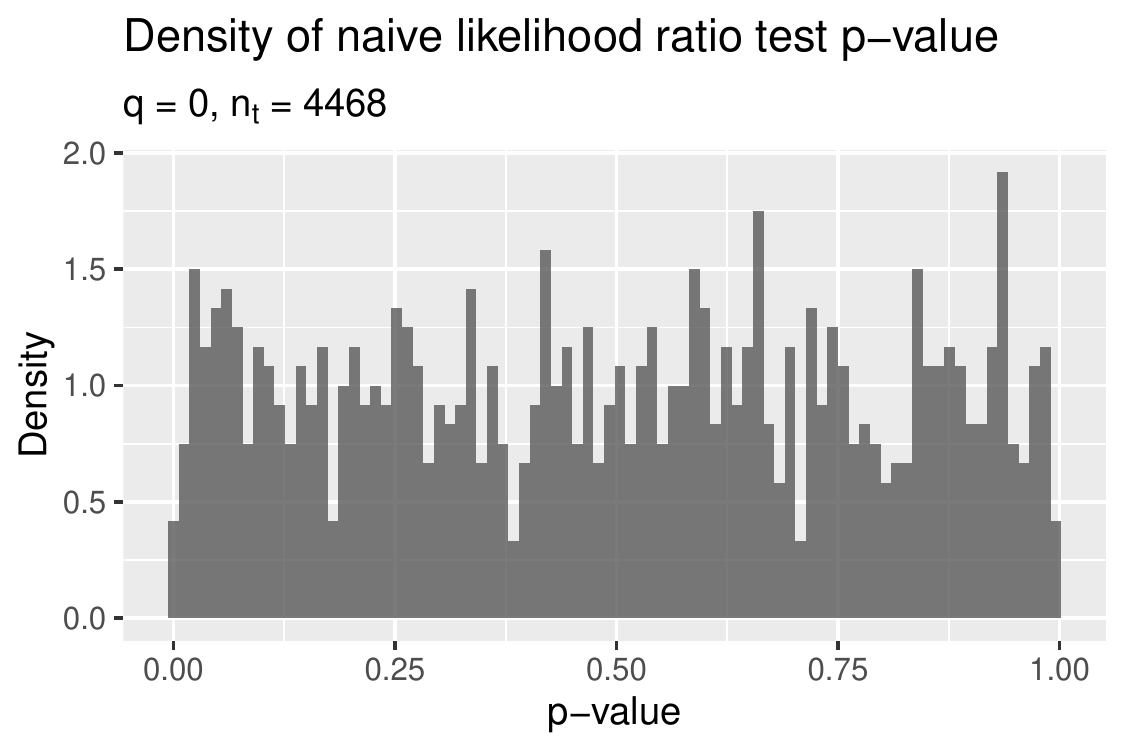}
\end{subfigure}
\begin{subfigure}[h]{0.49\linewidth}
\includegraphics[width=\linewidth]{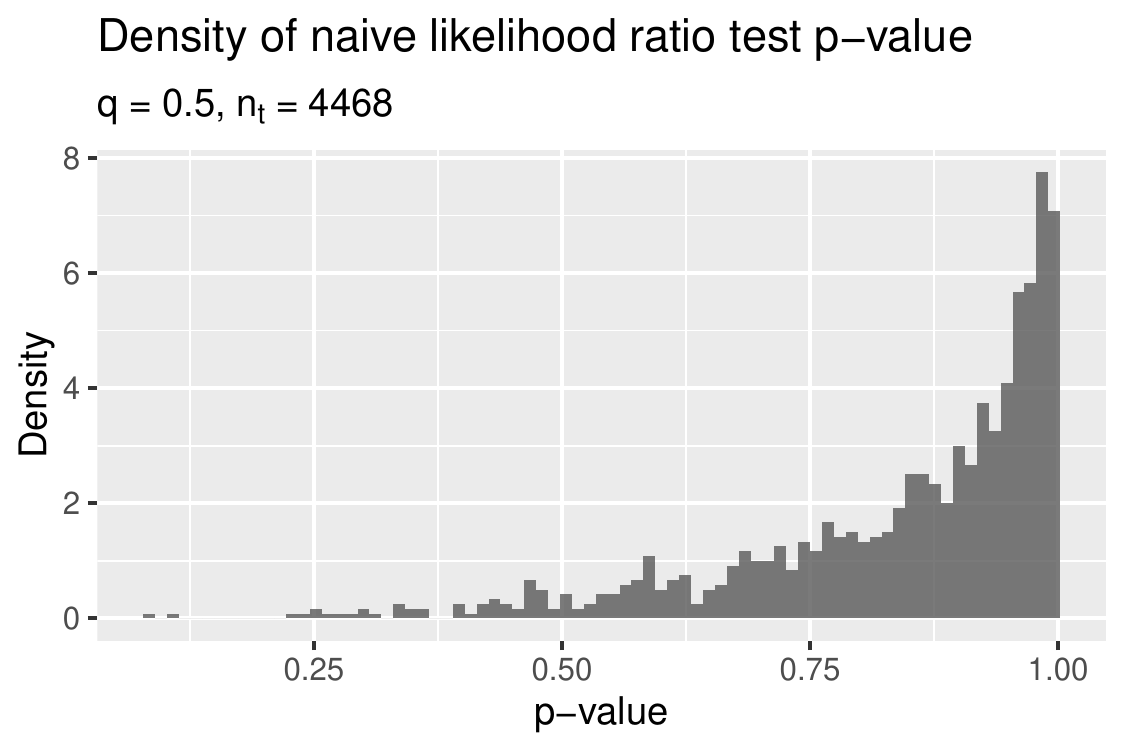}
\end{subfigure}
\end{center}
\vspace{-1.5em}
\caption{[Simulation study] Empirical distributions of naive likelihood ratio test p-values $\tilde{p}_{1,n}^{(r)}$ under $n_t=4468$ when $q=0$ (left panel) and $q=0.5$ (right panel).}
\label{fig:sim:hist_naive}
\end{figure}

\subsubsection{Correlation test} \label{sec:simu_corr}
To test the null hypothesis of (\ref{null_indep}) only for $T=2$ years, one may alternatively apply a \texttt{cor.test} function in \texttt{R} directly on the two empirical residuals $(Z_{i,1}-\hat{p}_{i,1})/\sqrt{\hat{p}_{i,1}(1-\hat{p}_{i,1})}$ and $(Z_{i,2}-\hat{p}_{i,2})/\sqrt{\hat{p}_{i,2}(1-\hat{p}_{i,2})}$, which is a standard t-test constructed based on the empirical estimate of residual correlation $\tilde{\rho}_{1,2}$ and its standard error $\text{SE}(\tilde{\rho}_{1,2})$. In this analysis, we compare $\text{SE}(\tilde{\rho}_{1,2})$ and the standard error produced by the proposed correlation test $\text{SE}(\hat{\rho}_{1,2})=\sqrt{\hat{\lambda}_{1,2}/n}$ to the true estimation uncertainty of the correlation, where $\hat{\lambda}_{1,2}$ is given by (\ref{eq:lambda_hat}). In this experiment, we consider $n_t=4468$ and $q=0.5$. Replicating the simulation by $R=1000$ times, we present the empirical distributions of $\text{SE}(\tilde{\rho}_{1,2})$ and $\text{SE}(\hat{\rho}_{1,2})$ in Figure \ref{fig:sim:hist_se}. The true estimation uncertainty of the correlation is approximated by the standard deviation of the estimated $\hat{\rho}_{1,2}$ (Equation (\ref{eq:corr_pair})) from the 1000 replications. We observe that $\text{SE}(\hat{\rho}_{1,2})$ under the proposed correlation test adheres closely to the true estimation uncertainty of 0.0219. Meanwhile, $\text{SE}(\tilde{\rho}_{1,2})$ under the standard t-test ranges from 0.0125 to 0.0135, significantly underestimating the true uncertainty. The main reason for underestimation is that the standard t-test fails to take the estimation uncertainty of $\hat{p}_{i,t}$ into account when $\text{SE}(\tilde{\rho}_{1,2})$, which depends on $\hat{p}_{i,t}$, is computed.

\begin{figure}[!h]
\begin{center}
\begin{subfigure}[h]{0.49\linewidth}
\includegraphics[width=\linewidth]{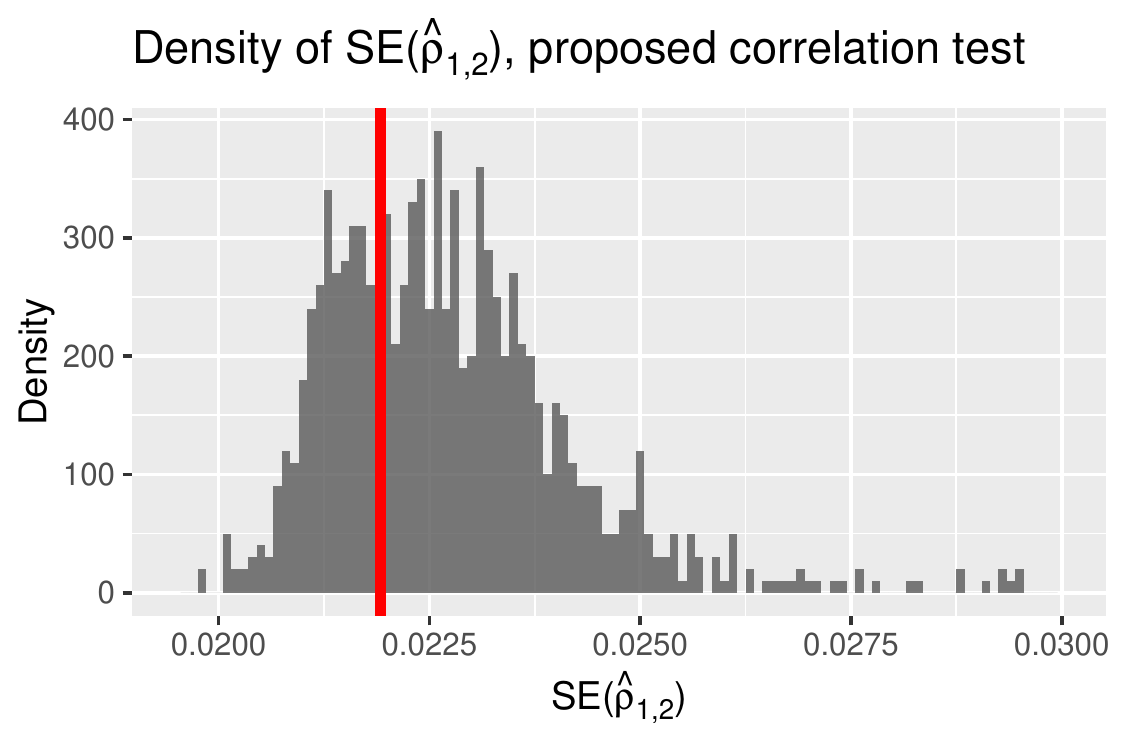}
\end{subfigure}
\begin{subfigure}[h]{0.49\linewidth}
\includegraphics[width=\linewidth]{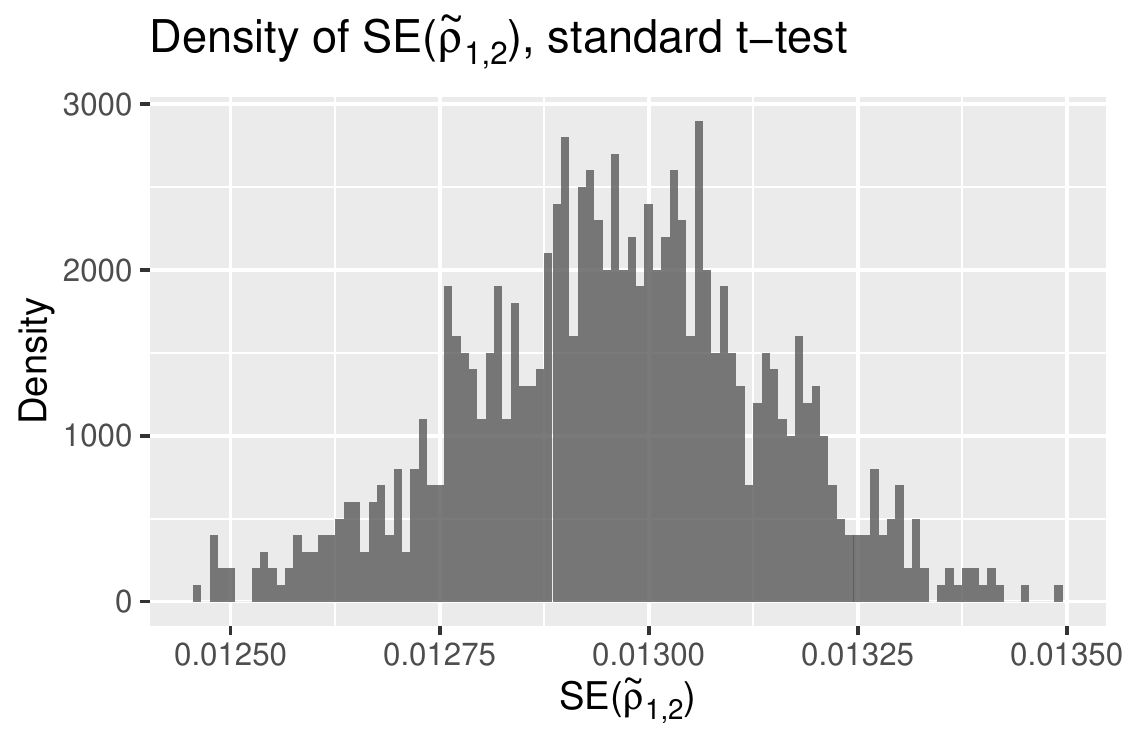}
\end{subfigure}
\end{center}
\vspace{-1.5em}
\caption{[Simulation study] Empirical distributions of the standard error of the residual correlation under the proposed correlation test ($\text{SE}(\hat{\rho}_{1,2})$, left panel) and the standard t-test using the \texttt{cor.test} function ($\text{SE}(\tilde{\rho}_{1,2})$, right panel). The vertical solid line in the left panel is the true uncertainty of the residual correlation.}
\label{fig:sim:hist_se}
\end{figure}

\section{Conclusion}\label{sec:conclusion}
A recent research interest in the actuarial literature is to model non-life insurance data over multiple years. Because of a smaller number of years, a constant time series dynamic is often imposed, implying that using longitudinal actuarial data is not necessary for forecast but useful in improving inference efficiency when the serial correlation is well captured.  In this paper, we develop two diagnostic tests, the \emph{serial dynamic test} and the \emph{correlation test}, to assess the assumptions. The serial dynamic test mainly detects non-constant serial dynamics of the conditional claim probability over time, which can hardly be validated statistically because of a short observation period. The correlation test evaluates the conditional independence of the claim 
probability over time, guiding in choosing an appropriate serial dependence structure for the longitudinal actuarial data.

The proposed methodologies are applied to two real insurance datasets. The serial dynamic test detects a structural change of the conditional claim probabilities over time for the LGPIF data but not the French motor data. The correlation test also reveals that the dependence structure of the LGPIF data may shift over time. Furthermore, the correlation test suggests using a short-range serial dependence model for the LGPIF dataset but a long-range dependence model for the French motor dataset. Overall, the French motor data generally does not violate the assumptions for the longitudinal models often made by the existing literature, justifying the use of the existing predictive models for longitudinal actuarial data to predict future claim distributions based on the last few years' data. The findings are the opposite for the LGPIF data, where the constant serial dynamic assumption is violated, and the use of the existing longitudinal models may result in misleading predictions.

This paper develops the diagnostic tests only based on the first step of the ratemaking process: the claim occurrence probabilities. Our future plan is to develop similar tests for risk measures, which are much more involved technically and numerically.
\section*{Acknowledgements}
We thank two reviewers for their helpful comments. Peng’s research was partly supported by the NSF grant of DMS-2012448. Qian's research was supported by the National Natural Science Foundation of China (12171158, 12071147, 12271171), the State Key Program of National Natural Science Foundation of China (71931004), Fundamental Research Funds for the Central Universities (2022QKT001), and the 111 Project (B14019).

\section{Proofs}\label{sec:proofs}
\begin{proof}[Proof of Theorem \ref{Th1}] 
Write
\[
	\begin{array}{ll}
	L_t(\boldsymbol{\gamma}_t)&=\sum_{i\in A_t}\{Z_{i,t}\log(p_{i,t})+(1-Z_{i,t})\log(1-p_{i,t})\}\\
	&=\sum_{i=1}^nI(i\in A_t)\{Z_{i,t}\log(p_{i,t})+(1-Z_{i,t})\log(1-p_{i,t})\}.
	\end{array}\]
Because
\[p_{i,t}=\frac{\exp(\alpha_t+\boldsymbol{\beta}_t^{\top}\boldsymbol{X}_{i,t})}{1+\exp(\alpha_t+\boldsymbol{\beta}_t^{\top}\boldsymbol{X}_{i,t})},\]
we have
\[
\left\{\begin{array}{ll}
&\frac{\partial p_{i,t}}{\partial \boldsymbol{\gamma}_t}=p_{i,t}\left( 1-p_{i,t}\right) \boldsymbol{\bar{X}}_{i,t},~
	\frac{\partial L_{t}(\boldsymbol{\gamma}_t)}{\partial \boldsymbol{\gamma}_t}=\sum_{i=1}^nI(i\in A_t)(Z_{i,t}-p_{i,t})\bar{\boldsymbol{X}}_{i,t},\\
&\frac{\partial^2 L_t(\boldsymbol{\gamma}_t)}{\partial \boldsymbol{\gamma}_t\partial \boldsymbol{\gamma}_t^\top}=-\sum_{i=1}^nI(i\in A_t)p_{i,t}\left( 1-p_{i,t}\right) 
	\boldsymbol{\bar{X}}_{i,t}\boldsymbol{\bar{X}}_{i,t}^\top.\end{array}\right.\]
For constant vectors $\boldsymbol{\lambda}_1,\cdots,\boldsymbol{\lambda}_T$ with the same dimension as $\bar{\boldsymbol{X}}_{i,t}$, using condition (C3), we have
\begin{equation}\label{pfTh1-2a}
E\{\sum_{t=1}^T\boldsymbol{\lambda}_t^{\top}\frac{\partial L_t(\boldsymbol{\gamma_t})}{\partial\boldsymbol{\gamma}_t}\}=\sum_{i=1}^n\sum_{t=1}^TI(i\in A_t)E\{\boldsymbol{\lambda}_t^{\top}\bar{\boldsymbol{X}}_{i,t}E(Z_{i,t}-p_{i,t}|\boldsymbol{X}_{i,t})\}=0,
\end{equation}
\begin{equation}\label{pfTh1-2b}
\begin{array}{ll}
&\frac 1n\sum_{i=1}^nE\{\sum_{t=1}^TI(i\in A_t)(Z_{i,t}-p_{i,t})\boldsymbol{\lambda}_t^{\top}\bar{\boldsymbol{X}}_{i,t}\}^2\\
=&\frac 1n\sum_{i=1}^n\sum_{t=1}^T\sum_{s=1}^TI(i\in A_t\cap A_s)E\{(Z_{i,t}-p_{i,t})(Z_{i,s}-p_{i,s})\boldsymbol{\lambda}_t^{\top}\bar{\boldsymbol{X}}_{i,t}\bar{\boldsymbol{X}}_{i,s}^{\top}\boldsymbol{\lambda}_s\}\\
=&\frac 1n\sum_{i=1}^n\sum_{t=1}^T\sum_{s=1}^TI(i\in A_t\cap A_s)E\{\boldsymbol{\lambda}_t^{\top}\bar{\boldsymbol{X}}_{i,t}\bar{\boldsymbol{X}}_{i,s}^{\top}\boldsymbol{\lambda}_sE((Z_{i,t}-p_{i,t})(Z_{i,s}-p_{i,s})|\boldsymbol{X}_{i,t}, \boldsymbol{X}_{i,s})\}\\
=&\frac 1n\sum_{i=1}^n\sum_{t=1}^T\sum_{s=1}^TI(i\in A_t\cap A_s)E\{\boldsymbol{\lambda}_t^{\top}\bar{\boldsymbol{X}}_{i,t}\bar{\boldsymbol{X}}_{i,s}^{\top}\boldsymbol{\lambda}_s(p_{i,t,s}-p_{i,t}p_{i,s})\}\\
=&\frac 1n\sum_{t=1}^T\sum_{s=1}^Tn_{t,s}E\{\boldsymbol{\lambda}_t^{\top}\bar{\boldsymbol{X}}_{1,t}\bar{\boldsymbol{X}}_{1,s}^{\top}\boldsymbol{\lambda}_s(p_{1,t,s}-p_{1,t}p_{1,s})\}\\
\to&\sum_{t=1}^T\sum_{s=1}^Ta_{t,s}\boldsymbol{\lambda}_t^{\top}E\{(p_{1,t,s}-p_{1,t}p_{1,s})\bar{\boldsymbol{X}}_{1,t}\bar{\boldsymbol{X}}_{1,s}^{\top}\}\boldsymbol{\lambda}_s
\end{array}\end{equation}
and
\begin{equation}\label{pfTh1-2c}\sum_{i=1}^nE\{|\sum_{t=1}^TI(i\in A_t)(Z_{i,t}-p_{i,t})\boldsymbol{\lambda}_t^{\top}\bar{\boldsymbol{X}}_{i,t}|^{2+\delta}\}=O(n)=o(n^{1+\frac{\delta}{2}}).
\end{equation}
It follows from (\ref{pfTh1-2a}), (\ref{pfTh1-2b}), (\ref{pfTh1-2c}), and the central limit theorem for a sum of independent but not identically distributed random variables (see the Corollary in Page 30 of 	\cite{Serfling2002}) that
\begin{equation}\label{pfTh1-2}
\frac 1{\sqrt n}\sum_{t=1}^T\boldsymbol{\lambda}_t^{\top}\frac{\partial L_t(\boldsymbol{\gamma}_t)}{\partial\boldsymbol{\gamma}_t}\overset{d}{\to}N(0, \sum_{t=1}^T\sum_{s=1}^Ta_{t,s}\boldsymbol{\lambda}_t^{\top}E\{(p_{1,t,s}-p_{1,t}p_{1,s})\bar{\boldsymbol{X}}_{1,t}\bar{\boldsymbol{X}}_{1,s}^{\top}\}\boldsymbol{\lambda}_s).\end{equation}
 By (\ref{pfTh1-2}) and the Cram\'er-Wold theorem, we have
 \begin{equation}
 \label{pfTh1-3}
 \frac 1{\sqrt n}(\frac{\partial L_1(\boldsymbol{\gamma}_1)}{\partial\boldsymbol{\gamma}_1^{\top}},\cdots,\frac{\partial L_T(\boldsymbol{\gamma}_T)}{\partial\boldsymbol{\gamma}_T^{\top}})^{\top}\overset{d}{\to} (\boldsymbol{W}_1^{\top},\cdots,\boldsymbol{W}_T^{\top})^{\top},
 \end{equation}
where $\boldsymbol{W}_t$'s are defined in Theorem \ref{Th1}.
 It follows from the law of large numbers for a sum of independent variables that
\begin{equation}\label{pfTh1-4}
\frac 1{n}	\frac{\partial^2 L_t(\boldsymbol{\gamma}_t)}{\partial \boldsymbol{\gamma}_t\partial \boldsymbol{\gamma}_t^\top}\overset{p}{\to}-a_t\Sigma_t.
\end{equation}
Therefore, it follows from the Taylor expansion that
\[
\boldsymbol{0}=\frac 1{\sqrt n}\frac{\partial L_t(\hat{\boldsymbol{\gamma}}_t)}{\partial\boldsymbol{\gamma}_t}=\frac 1{\sqrt n}\frac{\partial L_t({\boldsymbol{\gamma}}_t)}{\partial\boldsymbol{\gamma}_t}+\frac 1n\frac{\partial^2L_t(\boldsymbol{\gamma}_t)}{\partial\boldsymbol{\gamma}_t\partial\boldsymbol{\gamma}_t^{\top}}\sqrt n(\hat{\boldsymbol{\gamma}}_t-\boldsymbol{\gamma}_t)+o_p(1),
\]
i.e.,
\begin{equation}\label{pfTh1-5}
\sqrt n(\hat{\boldsymbol{\gamma}}_t-\boldsymbol{\gamma}_t)=-\{\frac 1n\frac{\partial^2L_t(\boldsymbol{\gamma}_t)}{\partial\boldsymbol{\gamma}_t\partial\boldsymbol{\gamma}_t^{\top}}\}^{-1}\frac 1{\sqrt n}\frac{\partial L_t({\boldsymbol{\gamma}}_t)}{\partial\boldsymbol{\gamma}_t}+o_p(1),\end{equation}
implying that
	\[\begin{array}{ll}
	&\sqrt{n}(\hat{\boldsymbol{\gamma}}_1^{\top}-\boldsymbol{\gamma}_1^{\top},\cdots,\hat{\boldsymbol{\gamma}}_T^{\top}-\boldsymbol{\gamma}_T^{\top})^{\top}\\
	=&\left((a_1^{-1}\Sigma_1^{-1}\frac 1{\sqrt n}\frac{\partial L_1(\boldsymbol{\gamma}_1)}{\partial\boldsymbol{\gamma}_1})^{\top},\cdots, (a_T^{-1}\Sigma_T^{-1}\frac{\partial L_T(\boldsymbol{\gamma}_T)}{\partial\boldsymbol{\gamma}_T})^{\top}\right)^{\top}+o_p(1)\\	
	\overset{d}{\to}&\left(a_1^{-1}\left( \Sigma_1^{-1}\boldsymbol{W}_1\right) ^{\top},\cdots, a_T^{-1}\left( \Sigma_T^{-1}\boldsymbol{W}_T\right) ^{\top}\right)^{\top}\end{array}\]
by  (\ref{pfTh1-3}) and (\ref{pfTh1-4}). Hence,  the theorem follows from the delta method.
\end{proof}

\begin{proof}[Proof of Theorem \ref{Th2}]
Write
\[L_t^b(\boldsymbol{\gamma}_t)=\sum_{i=1}^n\delta_i^bI(i\in A_t)\{Z_{i,t}\log(p_{i,t})+(1-Z_{i,t})\log(1-p_{i,t})\}.\]
Then, we have
\[
\left\{\begin{array}{ll}
	&\frac{\partial L_{t}^b(\boldsymbol{\gamma}_t)}{\partial \boldsymbol{\gamma}_t}-\frac{\partial L_{t}(\boldsymbol{\gamma}_t)}{\partial \boldsymbol{\gamma}_t}=\sum_{i=1}^n(\delta_i^b-1)I(i\in A_t)(Z_{i,t}-p_{i,t})\bar{\boldsymbol{X}}_{i,t},\\
&	\frac{\partial^2 L_t^b(\boldsymbol{\gamma}_t)}{\partial \boldsymbol{\gamma}_t\partial \boldsymbol{\gamma}_t^\top}=-\sum_{i=1}^n\delta_i^bI(i\in A_t)p_{i,t}\left( 1-p_{i,t}\right) 
	\boldsymbol{\bar{X}}_{i,t}\boldsymbol{\bar{X}}_{i,t}^\top.\end{array}\right.\]
Because $E(\delta_i^b)=1$, $E(\delta_i^b-1)^2=1$, and $\{\delta_i^b\}$ is independent of $Z_{i,t}$'s and $\boldsymbol{X}_{i,t}$'s, similar to (\ref{pfTh1-3}) and (\ref{pfTh1-4}), we have
\begin{equation}\label{pfTh2-1}
\frac 1n\frac{\partial^2 L_t^b(\boldsymbol{\gamma}_t)}{\partial\boldsymbol{\gamma}_t\partial\boldsymbol{\gamma}_t^{\top}}=-a_t\Sigma_t+o_p(1)
\end{equation}
and
\begin{equation}\label{pfTh2-2}
\frac 1{\sqrt n}(\frac{\partial L_1^b(\boldsymbol{\gamma}_1)}{\partial\boldsymbol{\gamma}_1^{\top}}-\frac{\partial L_1(\boldsymbol{\gamma}_1)}{\partial\boldsymbol{\gamma}_1^{\top}}, \cdots, \frac{\partial L_T^b(\boldsymbol{\gamma}_T)}{\partial\boldsymbol{\gamma}_T^{\top}}-\frac{\partial L_T(\boldsymbol{\gamma}_T)}{\partial\boldsymbol{\gamma}_T^{\top}})^{\top}\overset{d}{\to}(\boldsymbol{W}_1^{b\top},\cdots,\boldsymbol{W}_T^{b\top})^{\top},\end{equation}
which is independent of and has the same distribution as $(\boldsymbol{W}_1^{\top},\cdots,\boldsymbol{W}_T^{\top})^{\top}$ defined in Theorem \ref{Th1}.
Like the proof of (\ref{pfTh1-5}), the Taylor expansion  and (\ref{pfTh2-1}) yield
\begin{equation}\label{pfTh2-3}
\begin{array}{ll}
\sqrt n(\hat{\boldsymbol{\gamma}}_t^b-\boldsymbol{\gamma}_t)&=-\{\frac 1n\frac{\partial^2L_t^b(\boldsymbol{\gamma}_t)}{\partial\boldsymbol{\gamma}_t\partial\boldsymbol{\gamma}_t^{\top}}\}^{-1}\frac 1{\sqrt n}\frac{\partial L_t^b({\boldsymbol{\gamma}}_t)}{\partial\boldsymbol{\gamma}_t}+o_p(1)\\
&= a_t^{-1}\Sigma_t ^{-1}\frac 1{\sqrt n}\frac{\partial L_t^b({\boldsymbol{\gamma}}_t)}{\partial\boldsymbol{\gamma}_t}+o_p(1).\end{array}\end{equation}
By (\ref{pfTh1-5}) and (\ref{pfTh2-3}), we have
\begin{equation}\label{pfTh2-4}
\sqrt n(\hat{\boldsymbol{\gamma}}_t^b-\hat{\boldsymbol{\gamma}}_t)
=a_t^{-1}\Sigma_t^{-1}\{\frac 1{\sqrt n}\frac{\partial L_t^b({\boldsymbol{\gamma}}_t)}{\partial\boldsymbol{\gamma}_t}-\frac 1{\sqrt n}\frac{\partial L_t({\boldsymbol{\gamma}}_t)}{\partial\boldsymbol{\gamma}_t}\}+o_p(1).\end{equation}
Hence, it follows  (\ref{pfTh2-2}), and (\ref{pfTh2-4}) that the joint limit of $\sqrt n(\hat{\boldsymbol{\gamma}}_1^b-\hat{\boldsymbol{\gamma}}_1), \cdots, \sqrt n(\hat{\boldsymbol{\gamma}}_T^b-\hat{\boldsymbol{\gamma}}_T)$ is the same as that of
$\sqrt n(\hat{\boldsymbol{\gamma}}_1-\boldsymbol{\gamma}_1), \cdots, \sqrt n(\hat{\boldsymbol{\gamma}}_T-\boldsymbol{\gamma}_T)$, i.e., the theorem holds.
\end{proof}

\begin{proof}[Proof of Theorem \ref{Th3}]
It directly follows from Theorems \ref{Th1} and \ref{Th2}.
\end{proof}

\begin{proof}[Proof of Theorem \ref{Th4}]
	Define $f(x,y)=\frac{y-x}{\sqrt{x(1-x)}}$.  Then
	\[
		\frac{\partial f(x,y)}{\partial x}=\frac{2xy-x-y}{2\left\lbrace x(1-x)\right\rbrace ^{\frac{3}{2}}}~\text{and}~\frac{\partial f(p_{i,t},Z_{i,t})}{\partial \bm{\gamma_t}}=\frac{2Z_{i,t}p_{i,t}-Z_{i,t}-p_{i,t}}{2\sqrt{ p_{i,t}(1-p_{i,t})} }\bar{\boldsymbol{X}}_{i,t}.
	 \]
We only derive the limit of $\hat{\rho}_{s,t}$ for any $s<t$ as other cases can be done similarly.
An expansion of $\sqrt{n}{\hat{\rho}_{s,t}}$ at $(\boldsymbol{\gamma}_s^\top,\boldsymbol{\gamma}_t^\top)^\top$ gives
	\begin{align}\label{cor-Talyor}
		\sqrt{n}{\hat{\rho}_{s,t}}&=\sqrt{n} \frac{1}{n_{s,t}}\sum_{i\in A_{s}\cap A_t}f(p_{i,s},Z_{i,s})f(p_{i,t},Z_{i,t})\nonumber\\
		&\quad+\frac{1}{n_{s,t}}\sum_{i\in A_{s}\cap A_t}\frac{\partial f(p_{i,s},Z_{i,s})}{\partial \bm{\gamma}_s^\top}f(p_{i,t},Z_{i,t})\sqrt{n}\left( \hat{\bm{\gamma}}_s-\bm{\gamma}_s\right)\nonumber\\
		&\quad+\frac{1}{n_{s,t}}\sum_{i\in A_{s}\cap A_t}\frac{\partial f(p_{i,t},Z_{i,t})}{\partial \bm{\gamma}_t^\top}f(p_{i,s},Z_{i,s})\sqrt{n}\left( \hat{\bm{\gamma}}_t-\bm{\gamma}_t\right)+o_p(1).
	\end{align}
	Because
	\begin{align*}
		&E||\frac{2Z_{i,t}p_{i,t}-Z_{i,t}-p_{i,t}}{2\sqrt{p_{i,t}(1-p_{i,t})}}\frac{Z_{i,s}-p_{i,s}}{\sqrt{ p_{i,s}(1-p_{i,s})}}\bar{\boldsymbol{X}}_{i,t}||\\
		&=\frac{1}{2}E\left[ ||\bar{\boldsymbol{X}}_{i,t}|| E\left\lbrace \frac{|2Z_{i,t}p_{i,t}-Z_{i,t}-p_{i,t}|}{\sqrt{ p_{i,t}(1-p_{i,t})}}\frac{|Z_{i,s}-p_{i,s}|}{\sqrt{ p_{i,s}(1-p_{i,s})}}|\boldsymbol{X}_{i,s},\boldsymbol{X}_{i,t}\right\rbrace \right] \\
		&\leq \frac{1}{2}E\left[ ||\bar{\boldsymbol{X}}_{i,t}|| \sqrt{E\left\lbrace \frac{(2Z_{i,t}p_{i,t}-Z_{i,t}-p_{i,t})^2}{ p_{i,t}(1-p_{i,t})}|\boldsymbol{X}_{i,t}\right\rbrace E\left\lbrace \frac{(Z_{i,s}-p_{i,s})^2}{ p_{i,s}(1-p_{i,s}) }|\boldsymbol{X}_{i,s}\right\rbrace} \right] \\
		&\leq\frac{1}{2}E||\bar{\boldsymbol{X}}_{i,t}||<\infty,
	\end{align*}
	it follows from the weak law of large numbers and condition (C6) that
	\begin{equation*}
		\left\lbrace \begin{array}{ll}
			\frac{1}{n_{s,t}}\sum\limits_{i\in A_{s}\cap A_t}\frac{\partial f(p_{i,s},Z_{i,s})}{\partial \bm{\gamma}_s^{\top}}f(p_{i,t},Z_{i,t})=\omega_{s,t} +o_p(1),\\
			\frac{1}{n_{s,t}}\sum\limits_{i\in A_{s}\cap A_t}\frac{\partial f(p_{i,t},Z_{i,t})}{\partial \bm{\gamma}_t^{\top}}f(p_{i,s},Z_{i,s})=\omega_{t,s}+o_p(1),
		\end{array}\right.
	\end{equation*}
	where 
	$$\omega_{s,t}=E\left[ \frac{(p_{1,s,t}-p_{1,s}p_{1,t})(2p_{1,s}-1)}{2\sqrt{ p_{1,s}(1-p_{1,s})}\sqrt{ p_{1,t}(1-p_{1,t})}}\bar{\boldsymbol{X}}_{1,s}^\top\right] $$
	and
	$$\omega_{t,s}=E\left[ \frac{(p_{1,s,t}-p_{1,s}p_{1,t})(2p_{1,t}-1)}{2\sqrt{ p_{1,s}(1-p_{1,s})}\sqrt{ p_{1,t}(1-p_{1,t})}}\bar{\boldsymbol{X}}_{1,t}^\top\right].$$
	By (\ref{cor-Talyor}), Theorem \ref{Th1}, and the Slutsky theorem, we have
	\begin{align}\label{Cor}
		\sqrt{n}{\hat{\rho}_{s,t}}&=\frac{1}{\sqrt{n}}\sum_{i=1}^n \{ a_{s,t}^{-1}I(i\in  A_{s}\cap A_t)\frac{Z_{i,s}-p_{i,s}}{\sqrt{p_{i,s}(1-p_{i,s})}}\frac{Z_{i,t}-p_{i,t}}{\sqrt{p_{i,t}(1-p_{i,t})}} \nonumber\\
		&\quad+a_s^{-1}\omega_{s,t}\Sigma_s^{-1}I(i \in A_s) (Z_{i,s}-p_{i,s})\bar{\boldsymbol{X}}_{i,s}\nonumber\\
		&\quad+a_t^{-1}\omega_{t,s}\Sigma_t^{-1}I(i \in A_t) (Z_{i,t}-p_{i,t})\bar{\boldsymbol{X}}_{i,t}\}+o_p(1)\nonumber\\
		&:=\frac{1}{\sqrt{n}}\sum_{i=1}^n\left( J_{i,s,t}+J_{i,s}+J_{i,t}\right). 
	\end{align}
	Since \begin{equation}\label{cor-mean}
	E \sum_{i=1}^n\left( J_{i,s,t}+J_{i,s}+J_{i,t}\right)  =n_{s,t}a_{s,t}^{-1}\rho _{s,t},
		\end{equation}
	\begin{align}\label{cor-var}
		&\frac{1}{n} \sum_{i=1}^nE\left(  J_{i,s,t}+J_{i,s}+J_{i,t}-EJ_{i,s,t}\right)  ^2\nonumber\\&=\frac{1}{n}\sum_{i=1}^n \left[ \right.I(i\in  A_{s}\cap A_t)a_{s,t}^{-2}E\left\lbrace \frac{Z_{i,s}-p_{i,s}}{\sqrt{p_{i,s}(1-p_{i,s})}}\frac{Z_{i,t}-p_{i,t}}{\sqrt{p_{i,t}(1-p_{i,t}})} -\rho_{s,t}\right\rbrace^2 \nonumber\\
		&+I(i \in A_s)a_s^{-2}\omega_{s,t}\Sigma_s^{-1} E\left\lbrace (Z_{i,s}-p_{i,s})^2\bar{\boldsymbol{X}}_{i,s}\bar{\boldsymbol{X}}_{i,s}^\top\right\rbrace 
		\Sigma_s^{-1}\omega_{s,t}^\top\nonumber\\
		&+I(i \in A_t)a_t^{-2}\omega_{t,s}\Sigma_t^{-1} E\left\lbrace (Z_{i,t}-p_{i,t})^2\bar{\boldsymbol{X}}_{i,t}\bar{\boldsymbol{X}}_{i,t}^\top\right\rbrace \Sigma_t^{-1}\omega_{t,s}^\top\nonumber\\
		&+2I(i\in  A_{s}\cap A_t)a_{s,t}^{-1}a_s^{-1}\omega_{s,t}\Sigma_s^{-1}E\left\lbrace  \frac{(Z_{i,s}-p_{i,s})^2}{\sqrt{p_{i,s}(1-p_{i,s}})}\frac{Z_{i,t}-p_{i,t}}{\sqrt{p_{i,t}(1-p_{i,t}})} \bar{\boldsymbol{X}}_{i,s} \right\rbrace\nonumber\\
		&+2I(i\in  A_{s}\cap A_t)a_{s,t}^{-1}a_t^{-1}\omega_{t,s}\Sigma_t^{-1}E\left\lbrace  \frac{Z_{i,s}-p_{i,s}}{\sqrt{p_{i,s}(1-p_{i,s}})}\frac{(Z_{i,t}-p_{i,t})^2}{\sqrt{p_{i,t}(1-p_{i,t}})} \bar{\boldsymbol{X}}_{i,t} \right\rbrace\nonumber\\
		&+2I(i \in A_s \cap A_t)a_s^{-1}a_t^{-1}\omega_{s,t}\Sigma_s^{-1} E\left\lbrace (Z_{i,s}-p_{i,s})(Z_{i,t}-p_{i,t})\bar{\boldsymbol{X}}_{i,s}\bar{\boldsymbol{X}}_{i,t}^\top\right\rbrace\Sigma_t^{-1}\omega_{t,s}^\top \left.\right]\nonumber \\
		&\overset{p}\to  a_{s,t}^{-1}E\left\lbrace \frac{(p_{1,s,t}-p_{1,s}p_{1,t})(1-2p_{1,s})(1-2p_{1,t})}{p_{1,s}(1-p_{1,s})p_{1,t}(1-p_{1,t})}\right\rbrace  +a_{s,t}^{-1}(1-\rho_{s,t}^2)-3a_{s}^{-1}\omega_{s,t}\Sigma_s^{-1}\omega_{s,t}^\top\nonumber\\
		&-3a_t^{-1}\omega_{t,s}\Sigma_t^{-1} \omega_{t,s}^\top
		+2a_{s,t}a_s^{-1}a_t^{-1}\omega_{s,t}\Sigma_s^{-1} E\left\lbrace (p_{1,s,t}-p_{1,s}p_{1,t})\bar{\boldsymbol{X}}_{1,s}\bar{\boldsymbol{X}}_{1,t}^\top\right\rbrace\Sigma_t^{-1}\omega_{t,s}^\top:=\lambda_{s,t},
	\end{align}
and it follows from conditions (C3) and (C7) that
\begin{align}\label{cor-Lya}
&\sum_{i=1}^nE|J_{i,s,t}+J_{i,s}+J_{i,t}-EJ_{i,s,t} |^{2+\delta}\nonumber\\
&=\sum_{i=1}^nO(1)\times E\left\{\left[\frac{Z_{i,s}-p_{i,s}}{\sqrt{p_{i,s}(1-p_{i,s})}}\frac{Z_{i,t}-p_{i,t}}{\sqrt{p_{i,t}(1-p_{i,t})}}\right]^{2+\delta}\right\}\nonumber\\
&\leq\sum_{i=1}^nO(1)\times E\Bigg\{\left(\frac{1-p_{i,s}}{p_{i,s}}\frac{1-p_{i,t}}{p_{i,t}}\right)^{1+\frac{\delta}{2}}+\left(\frac{1-p_{i,s}}{p_{i,s}}\frac{p_{i,t}}{1-p_{i,t}}\right)^{1+\frac{\delta}{2}}\nonumber\\
&\hspace{8em}+\left(\frac{p_{i,s}}{1-p_{i,s}}\frac{1-p_{i,t}}{p_{i,t}}\right)^{1+\frac{\delta}{2}}+\left(\frac{p_{i,s}}{1-p_{i,s}}\frac{p_{i,t}}{1-p_{i,t}}\right)^{1+\frac{\delta}{2}}\Bigg\}\nonumber\\
&=\sum_{i=1}^nO(1)\times E\Big\{e^{(1+\frac{\delta}{2})(-\bm{\gamma}_s^{\top}\bar{\bm{X}}_{i,s}-\bm{\gamma}_t^{\top}\bar{\bm{X}}_{i,t})}+e^{(1+\frac{\delta}{2})(-\bm{\gamma}_s^{\top}\bar{\bm{X}}_{i,s}+\bm{\gamma}_t^{\top}\bar{\bm{X}}_{i,t})}\nonumber\\
&\hspace{8em}+e^{(1+\frac{\delta}{2})(\bm{\gamma}_s^{\top}\bar{\bm{X}}_{i,s}-\bm{\gamma}_t^{\top}\bar{\bm{X}}_{i,t})}+e^{(1+\frac{\delta}{2})(\bm{\gamma}_s^{\top}\bar{\bm{X}}_{i,s}+\bm{\gamma}_t^{\top}\bar{\bm{X}}_{i,t})}\Big\}\nonumber\\
&=O(n)=o(n^{1+\frac{\delta}{2}}).
\end{align}
Then, by (\ref{Cor}), (\ref{cor-mean}), (\ref{cor-var}), (\ref{cor-Lya}), and the central limit theorem, we have
$$\sqrt{n}\left( \hat{\rho}_{s,t}-\rho_{s,t}\right)\overset{d}\to N(0,\lambda_{s,t}).$$
\end{proof}
\begin{proof}[Proof of Theorem \ref{Th5}]
It is similar to the proof of Theorem \ref{Th2}.
\end{proof}
\begin{proof}[Proof of Theorem \ref{Th6}]
	It immediately follows from Theorems \ref{Th4} and \ref{Th5}.
\end{proof}

\bibliographystyle{apalike}
\bibliography{reference}

\end{document}